\documentclass[10pt,twocolumn,twoside]{ieeetran}

\ifCLASSINFOpdf
\else
\fi
\usepackage[comma,numbers,square,sort&compress]{natbib}
\usepackage{algorithm} 
\usepackage{algorithmic} 
\usepackage{setspace}
\usepackage{arydshln}
\usepackage{multirow} 
\usepackage{xcolor}
\usepackage{setspace}
\usepackage{amsmath}
\newtheorem{theorem}{\textbf{Theorem}}
\newtheorem{lemma}{\textbf{Lemma}}
\newtheorem{example}{\textbf{Example}}
\newtheorem{corollary}{\textbf{Corollary}}
\newtheorem{remark}{\textbf{Remark}}
\newtheorem{definition}{\textbf{Definition}}
\newtheorem{problem}{\textbf{Problem}}
\newtheorem{proposition}{\textbf{Proposition}}
\usepackage{pmat}
\newenvironment{proof}{{\noindent{\bf \emph{Proof:}}}\quad}{\hfill $\square$\par}

\usepackage{multirow}
\usepackage{url}
\usepackage{subfigure}
\usepackage{fancyhdr}
\usepackage{amsmath}
\usepackage{multirow}
\usepackage{amssymb }
\usepackage{color}
\usepackage{graphics} 
\usepackage{graphicx}
\usepackage{geometry}
\usepackage{setspace}

\renewcommand{\baselinestretch}{0.99999} \normalsize

\begin{document}
%
%
%
%

\title{Generic Detectability and Isolability of Topology Failures in Networked Linear Systems} 

\author{Yuan Zhang, Yuanqing Xia$^*$, Jinhui Zhang, and Jun Shang
\thanks{ This work was supported in part by the  China Postdoctoral Innovative Talent Support Program (BX20200055), the National Natural Science Foundation of China (62003042), and the State Key Program of National Natural Science Foundation of China (61836001).

Yuan Zhang, Yuanqing Xia (corresponding author) and Jinhui Zhang are with the School of Automation, Beijing Institute of Technology, Beijing, China {(email: {\tt\small zhangyuan14@bit.edu.cn, xia\_yuanqing@bit.edu.cn, zhangjinh@bit.edu.cn}).} 

Jun Shang is with the Department of Electrical and Computer Engineering, University of Alberta, Edmonton, Canada T6G 1H9 (email:{\tt\small  jshang2@ualberta.ca}).
       } }
\pagestyle{empty} 
\maketitle
\thispagestyle{empty} 
\begin{abstract}This paper studies the possibility of detecting and isolating topology failures (including link failures and node failures) of a networked system from subsystem measurements, in which subsystems are of fixed high-order linear dynamics, and the exact interaction weights among them are unknown. We prove that in such class of networked systems with the same network topologies, the detectability and isolability of a given topology failure (set) are generic properties, indicating that it is the network topology that dominates the property of being detectable or isolable for a failure (set). We first give algebraic conditions for detectability and isolability of arbitrary parameter perturbations for a lumped plant, and then derive graph-theoretical necessary and sufficient conditions for generic detectability and isolability of topology failures for the networked systems. On the basis of these results, we consider the problems of deploying the smallest set of sensors for generic detectability and isolability. We reduce the associated sensor placement problems to the hitting set problems, which can be effectively solved by greedy algorithms with guaranteed approximation performances.
\end{abstract}
\begin{IEEEkeywords}
Failure detectability and isolability, generic property, graph theory, sensor placement, networked system
\end{IEEEkeywords}
\section{Introduction}
There exist many large-scale systems consisting of a large number of subsystems in the real world. These subsystems, usually geographically distributed, are interconnected through a network. Such systems are often called networked systems. Many critical infrastructures can be modeled as networked systems, such as power systems \cite{kundur1994power}, the Internet \cite{Albert2000Error}, wireless communication networks \cite{wood2002denial}, and transportation networks \cite{Jadbabaie2003Coordination}. The security and reliability of networked systems have aroused great concern from various aspects \cite{Albert2000Error,wood2002denial,Buldyrev2009Catastrophic,F.Pa2013Attack}. 

In networked systems, a common type of fault is the perturbation/variant of components of its network structure. For example, links may be blocked or removed, making signals unable to be transmitted normally, and nodes (agents) may not operate normally or even lose communications with their neighbors, leading to loss of system performances. Such type of structure variants can result from either the failure of network components (such as links or nodes), or denial-of-service attacks \cite{Albert2000Error,wood2002denial,Buldyrev2009Catastrophic,Pasqualetti2012ConsensusCI,De2015Input}. The failure of a set of links or nodes is collectively called {\emph{topology failure}} in this paper. Topology failures may have disastrous impacts on the security and normal functioning of a networked system. One example is the catastrophic power outage in southern Italy in 2003, which was reportedly caused by failures of some high voltage transmission lines \cite{Buldyrev2009Catastrophic}. Considering the possible catastrophic cascading consequences caused by topology failures, the timely detection and isolation have become particularly important \cite{F.Pa2013Attack}. 

Fault detection and isolation (FDI) have long been active in control community \cite{Chen1999RobustMF,Chi2015SensorPF,Massoumnia1986AGA,Commault2002ObserverbasedFD,Namburu2007DataDrivenMF,Zhou2020ReviewOD}. The main targets are to determine whether faults occur and to locate them. Many detection and isolation approaches have been proposed, including geometric theory based approaches \cite{Massoumnia1986AGA}, observer-based approaches \cite{Commault2002ObserverbasedFD}, data-driven approaches \cite{Namburu2007DataDrivenMF}, and so on.  However, the majority of literature on this topic deals with faults that are linked to either additional external signals or undesired parameter deviations \cite{Chen1999RobustMF,Zhou2020ReviewOD}. Topology failures, on the other hand, result in perturbations on the structure of system intrinsic dynamics. Unlike common parameter deviations, topology failures shift the nominal parameters to only some discrete values, which are usually hard to be modeled as external disturbances.

Nevertheless, in literature, the detection of topology failures has drawn on FDI techniques, that is, by comparing the discrepancies between the current system output and the nominal output to determine whether the system has undergone topology failures \cite{Chen1999RobustMF}. Such problems have recently attracted researchers' attention. In \cite{Aghdam2012Characterization,Aghdam2012Detectability},  Rahimian {\textsl{et~al.}} studied detectability of single or multiple link failures for multi-agent systems under the agreement protocol. They introduced the concept of distinguishable flow graph and gave sufficient conditions to distinguish faulty links. In \cite{Battistelli2015DetectingTV}, Battistelli and Tesi used mode observability from switching systems theory to characterize indiscernible states in networks of single-integrators, i.e., the initial states that generate exactly the same outputs for the nominal system and the system after failures. The same authors further extended the former work to networked diffusively coupled high-order systems \cite{BattistelliDetection2017}, whiles  Patil \textsl{et~al.} considered indiscernible topological variations in networks with descriptor subsystems, where the subsystems can be heterogeneous \cite{Patil2019IndiscernibleTV}. In \cite{Rahimian2015DetectionAI}, Rahimian and Preciado studied detection and isolation algorithms of single-link failure in networked linear systems. They related the discontinuity of higher-order derivatives of system outputs caused by the removal of a single link to the distance from the end of the removed link to the observed node.

However, all of the above works depend on accurate system parameters, which means accurate parameters are required when applied. In addition, it is usually not easy to extend their results to the case with simultaneous failures of multiple links/nodes. For many practical systems, accurate system parameters may be hard to obtain, but their zero-nonzero patterns, i.e., which entry of the system matrices is zero and which is not, might be easier accessible. This forms a class of systems sharing the same ``structure''. In control theory, some properties will become {\emph{generic}} in this class of systems, i.e., either for almost all systems in such class, these properties hold true, or for none these properties hold true. For example, controllability and observability are two well-known generic properties, both for a lumped structured plant \cite{generic} and a networked system with fixed subsystem dynamics and unknown subsystem interaction weights \cite{zhang2019structural,zhang2020structural}.

Generic properties are particularly prominent in analyzing large-scale networked systems, not only because they usually can intuitively show how topologies influence the considered properties, but also because they often can be verified efficiently by means of graphical tools \cite{generic,zhang2019structural}. In this paper, we study generic detectability and isolability of topology failures for a networked linear system, where subsystem dynamics are given and identical, but the weights of interaction links among them are unknown. We study under what conditions we can generically detect and isolate a given topology failure (set) from the nominal system dynamics and its output measurements. These conditions reveal fundamental limitations for the network topology to support detectability and isolability of a given failure (set), and are irrespective of the exact detection and isolation algorithms one adopts. Our main contributions are as follows.

1) We give algebraic conditions for detectability and isolability of topology failures.  These conditions are general enough in the sense that, they remain valid for arbitrary parameter perturbations in the system matrices not necessarily resulting from topology failures. 

2) We prove that detectability and isolability of topology failures for a networked system are both generic properties. That means,  it is how subsystems are interconnected, rather than the exact interaction weights, that dominates detectability and isolability of a given topology failure (set) for a class of networked systems sharing the same topologies. 

3) We give necessary and sufficient graph-theoretic conditions for generic detectability and isolability of a given topology failure (set). Compared to the existing literature \cite{Aghdam2012Characterization,Aghdam2012Detectability,BattistelliDetection2017}, these conditions are applicable to larger classes of topology failures, including single-link failure, single-node failure, or the failure of an arbitrary set of links. {Some characterizations of (non) generically isolable failure sets are also given.} Particularly, one interesting finding is that, the conditions for generic detectability of every single-link failure are equivalent to those for generic isolability of the set of all single-link failures. 

4) Finally, on the basis of the above results, we consider sensor placement problems aiming to using the minimal number of sensors to make a given failure (set) generically detectable (isolable). We reduce these problems to the hitting set problems, and use greedy algorithms to approximate them with guaranteed performances.

The rest of this paper is organized as follows. Section II gives problem formulations and some preliminaries. Section III provides algebraic conditions for detectability and isolability of topology failures for a lumped plant. Section IV demonstrates that the detectability and isolability are generic properties. Graph-theoretical conditions for the generic detectability and isolability of topology failures are given in Section {
\ref{fengwei}}. Sensor placement problems to achieve generic detectability and isolability are discussed in Section {\ref{section_sensor}}. In Section {\ref{section_example}},  some simulations and examples are provided to validate the theoretical results. The last section concludes this paper.

{\emph{Notations:}} $\mathbb R$, $\mathbb C$ and $\mathbb N$ denote the sets of real, complex and integer numbers, respectively.  For a set, $|\cdot |$ denotes its cardinality. For a matrix $M$, $M_{ij}$ or $[M]_{ij}$ denotes the entry in the $i$th row and $j$th column of $M$, and ${{\bf { ker}}(M)}$ denotes the null space of $M$. By ${\bf diag}\{X_i|_{i=1}^n\}$ we denote the block diagonal matrix whose $i$th diagonal block is $X_i$, and ${\bf col}\{X_i|_{i=1}^n\}$ the matrix stacked by $X_i|_{i=1}^n$. By $e_i^{[N]}$ we denote the $i$th column of the $N$ dimensional identify matrix $I_N$, and $e_{ij}^{[N]}$ the $N\times N$ matrix whose $(i,j)$th entry is one and the rest are zero. Symbol ${\bf abs}(\cdot)$ takes the absolute value of a scalar, and $A\otimes B$ denotes the Kronecker product of matrices $A$ and $B$.  For a square matrix $M$, $\rho(M)$ denotes its spectral radius, namely, the maximum absolute value of its eigenvalues.

\section{Problem Formulation and Preliminaries}\label{model_description}
\subsection{Preliminaries}
Concepts in graph theory:   In a directed graph (digraph) ${\cal D}=({\cal V},{\cal E})$, where $\cal V$ is the node set and ${\cal E} \subseteq {\cal V}\times {\cal V}$ is the edge (or link) set, a path from $v_i\in {\cal V}$ to $v_j\in {\cal V}$ is a sequence of edges $\{(v_{i}, v_{i+1}),(v_{i+1},v_{i+2}), ..., (v_{j-1}, v_{j})\}$. The length of a path is the number of edges it contains. The distance from $v_i$ to $v_j$ in $\cal D$, denoted by ${\rm dist}(v_i,v_j,{\cal D})$, is the length of the shortest path from $v_i$ to $v_j$. If there is no path from $v_i$ to $v_j$, then ${\rm dist}(v_i,v_j,{\cal D})=\infty$. In this paper, adjacency matrix of a weighted digraph $\cal D$ is a matrix $W\in {\mathbb R}^{|{\cal V}|\times |{\cal V}|}$ such that $W_{ij}\ne 0$ only if $(v_j,v_i)\in {\cal E}$, where $W_{ij}$ is the weight of $(v_j,v_i)$, the edge from $v_j$ to $v_i$. {\footnote{{It is worth noting that this definition is a little different from the conventional one, where $W_{ij}$ corresponds to the edge $(v_i,v_j)$.}}}

\subsection{Detectability and Isolability of Topology Failures} \label{problem}
Consider a networked system consisting of $N$ linear time invariant subsystems. Let ${\cal G}=({\cal V}, {\cal E})$ be a digraph describing the subsystem interconnection topology, with the node set ${\cal V}=\{1,...,N\}$, and a directed edge $(i,j)\in {\cal E}_{\rm sys}$ from node $i$ to node $j$ exists if the $j$th subsystem is directly influenced by the $i$th one. Dynamics of the $i$th subsystem is \footnote{In this paper, we focus on how the network topology plays its role in failure detectability and isolability. Hence, we do not take the external inputs into consideration (i.e., the external inputs are fixed to be zero). However, our approaches can be extended to the case with {\emph{known}} external inputs.}
\begin{equation} \label{sub_dynamic}
\dot x_i(t)= Ax_i(t) + B\sum\limits_{j=1}^Nw_{ij}\Gamma x_j(t),\\~
y_i(t)=Cx_i(t)
\end{equation}
where $A\in {\mathbb R}^{n\times n}$ is the state transition matrix, $B\in {\mathbb R}^{n\times m}$ is the input matrix, $\Gamma\in {\mathbb R}^{m\times n}$ is the internal coupling matrix between subsystems, $x_i(t)\in {\mathbb R}^{n}$ is the state vector, $y_i(t)\in {\mathbb R}^p$ is the subsystem output vector, and $w_{ij}\in {\mathbb R}$ is the weight of edge (link) from the $j$th subsystem to the $i$th one satisfying $w_{ij} \ne 0$ only if $(j,i)\in {\cal E}$, for $i,j\in\{1,...,N\}$. Denote the set of all weights $w_{ij}$ by $\{w_{ij}\}$. Notice that self-loops could be contained in $\cal E$, which could result from self-feedbacks, consensus-based agreement protocols, etc. %

Suppose that subsystems indexed by the set ${\cal S}\subseteq \{1,...,N\}$ are directly measured. Define
$$S\doteq {\bf col}\{[{e^{[N]}_i}]^{\intercal}|_{i\in {\cal S}}\}.$$ Let $x(t)=[x_1^{\intercal}(t),...,x_N^{\intercal}(t)]^{\intercal}$, $y(t)={\bf col}\{y_i(t)|_{i\in {\cal S}}\}$, and $W=[w_{ij}]$ be the adjacency matrix of $\cal G$. The lumped state-space representation of (\ref{sub_dynamic}) then is
\begin{equation} \label{lump_ss} \begin{aligned}
\dot x(t)=\Phi x(t), y(t)=Q x(t)\end{aligned}
\end{equation} where
\begin{equation} \label{lump_pp} \begin{array}{l}
 \Phi =I_N\otimes A+ W \otimes H, Q=S\otimes C,\end{array} \end{equation}with $H\doteq B\Gamma\in {\mathbb R}^{n\times n}$. Let $n_x\doteq Nn, n_y\doteq |{\cal S}|p$, then $\Phi\in {\mathbb R}^{n_x\times n_x}$, $Q\in {\mathbb R}^{n_y\times n_x}$.

Equation (\ref{sub_dynamic}) models a networked system with multi-input-multi-output subsystems, which arises in modeling interacted liquid tanks \cite{Modern_Control_Ogata}, synchronizing networks of linear oscillators \cite{LucaSynchronization,zhang2019structural}, electrical systems \cite{tuna2017observability}, power networks \cite{kundur1994power}, etc. 

In practical engineering, common topology failures include link failures and node (or agent) failures. The failure of a set of links ${\cal E}_f\subseteq {\cal E}$ corresponds to that all edges in ${\cal E}_f$ are removed from $\cal G$.  The failure of a set of nodes ${\cal V}_f\subseteq {\cal V}$ corresponds to that, for each node $i\in {\cal V}_f$, all edges adjacent to $i$,  i.e., $\{(i,j):(i,j)\in {\cal E}\}\bigcup\{(j,i):(j,i)\in {\cal E}\}$, are removed from $\cal G$. Obviously, node failures are special cases of link failures. Hence, we shall focus on link failures in the rest of this paper, and we will use the link set ${\cal E}_f\subseteq {\cal E}$ to denote the failure of removing all links of ${\cal E}_f$ from $\cal G$. With the failure ${\cal E}_f$,  the topology of the resulting networked system becomes ${\bar {\cal G}}=({\cal V}, {\cal E}\backslash{\cal E}_{f})$, with its adjacency matrix being denoted by $\bar W$, which is obtained from $W$ by setting the entries corresponding to ${\cal E}_f$ to zero.  We express  dynamics of (\ref{sub_dynamic}) after failure ${\cal E}_f$ as
\begin{equation} \label{lumpfaulty} \begin{aligned} \dot x(t)=\bar \Phi x(t), y(t)=Q x(t)\end{aligned} \end{equation}
with $\bar \Phi=I_N\otimes A+ \bar W \otimes H$.

The above formulation rises an interesting problem: Is it possible to detect and isolate topology failures from system outputs given the faultless nominal network dynamics (\ref{sub_dynamic}) ?  Let $y(x_0,{\cal G},t)$ (respectively, $y(x_0,\bar {\cal G}, t)$) be the output vector of networked system (\ref{sub_dynamic}) with topology ${\cal G}$  ($\bar {\cal G}$) and initial state $x_0$ at time $t\ge 0$. Following \cite{Aghdam2012Characterization,Aghdam2012Detectability,BattistelliDetection2017}, the detectability of failure ${\cal E}_f$ is defined as follows.

\begin{definition} \label{defdection} For networked system (\ref{sub_dynamic}), a failure ${\cal E}_f \subseteq {\cal E}$ is detectable if there exists an initial state $x_0\in {\mathbb R}^{n_x}$ such that $y(x_0,{\cal G},t)-y(x_0,{\bar {\cal G}},t)\not\equiv 0$, $t\ge 0$.
\end{definition}


{ \begin{remark} Definition \ref{defdection} concerns only the existence of an initial state that induces different outputs (see \cite{Ding2008ModelbasedFD,Aghdam2012Characterization,Aghdam2012Detectability,Rahimian2015DetectionAI} for similar definitions). In other words, not every initial state satisfies the inequality in this definition (in fact, the zero state $x_0=0_{n_x}$ always does not), and the feasible initial state $x_0$ may need to be chosen (or possibly known) for the detection implementation. It is also safe to say that Definition \ref{defdection} presents the minimal conditions required to detect topology failures from system output measurements. This may contract with similar definitions in the attack detection (c.f. \cite{F.Pa2013Attack}), which requires that every element in the attack set should be detectable because of the stealthiness behaviour of the attacks.
%

\end{remark}}


In the failure isolation problem, it is often the case that the exact failure is not known, but we may have prior knowledge of the possible failure candidates \cite{Chi2015SensorPF}. Suppose that the~emerging~failure belongs to a known prior topology failure set ${\mathbb{E}}=\{{\cal E}_1,...,{\cal E}_r\}$, where ${\cal E}_i\subseteq {\cal E}$, and $r$ is finite. For example, if at most one link is removed (namely, single-link failure), then ${\mathbb{E}}={\cal E}$. Since ${\mathbb{E}}$ is a combinatorial set of links in $\cal E$, we have $|{\mathbb{E}}|\le 2^{|\cal E|}$. Let ${\cal E}_0=\emptyset$. For each ${\cal E}_i$, $0\le i \le r$, let ${\cal G}_i=({\cal V}, {\cal E}\backslash {\cal E}_i)$. Failure isolation is possible from a prior failure set $\mathbb E$, only if there is a unique topology ${\cal G}_i$ that can explain the output response of the resulting networked system.

\begin{definition}\label{def1isolation} For networked system (\ref{sub_dynamic}), a failure set ${\mathbb{E}}=\{{\cal E}_1,...,{\cal E}_r\}$ is isolable if for any two integers $i,j\in \{0,...,r\}$ with $i\ne j$, there exists $x_{0ij}\in {\mathbb R}^{n_x}$ such that $y(x_{0ij},{\cal G}_i,t)-y(x_{0ij},{\cal G}_j,t)\not\equiv 0$.
\end{definition}

{ \begin{remark} The key difference between definitions of the topology identifiability in \cite{Waarde2019TopologyIO} and the failure isolability in Definition \ref{def1isolation} lies in that, the former requires that any changes in the nominal value of $W$ will affect the corresponding output response, whiles the latter only concerns the affections on the corresponding output response at some finite discrete perturbations of $W$ (this is also the key reason why the detectability and isolability studied in this paper are generic properties). As mentioned in \cite{BattistelliDetection2017}, such difference is due to the knowledge of system nominal dynamics and the prior failure set. \end{remark}}

We will show in the next section that, if a failure set ${\mathbb{E}}=\{{\cal E}_1,...,{\cal E}_r\}$ is isolable, then there exists a common $x_0\in {\mathbb R}^{n_x}$, such that $y(x_0,{\cal G}_i,t)-y(x_0,{\cal G}_j,t)\not\equiv 0$ for any two integers $i,j\in \{0,...,r\}$ with $i\ne j$.

In many practical scenarios, while parameters $A,B,\Gamma,C$ for subsystem dynamics are often known from physically modeling (one of the most common dynamics is the high-order integrator) or system identification,  the exact weights $\{w_{ij}\}$ among subsystems might be hard to know due to parameter uncertainties or geographical distance between subsystems. However, the knowledge about which $w_{ij}$ is zero or not may be easily accessible \cite{shahrampour2015topology,Waarde2019TopologyIO,zhang2019structural,zhang2020structural}. We will show failure detectability and isolability are generic properties. In other words, either for almost all weights $\{w_{ij}\}$ with the corresponding zero-nonzero patterns, a given failure (set) is detectable (isolable), or for all weights $\{w_{ij}\}$ with the corresponding zero-nonzero patterns, the answers to the same problems are NO.  The purpose of this paper is to find conditions under which such generic properties hold true, and apply them to the associated sensor placement problems.

\begin{remark}In literature, the assumption of knowing subsystem dynamics but with little/no knowledge on the subsystem interaction weights is common in many aspects on networked systems, including topology reconstruction \cite{shahrampour2015topology}, system identification \cite{Waarde2019TopologyIO}, as well as structural controllability \cite{zhang2019structural,zhang2020structural}.
\end{remark}

\section{Algebraic Conditions for Failure Detectability and Isolability} \label{algebraic_condition}
In this section, we will give necessary and sufficient algebraic conditions for failure detectability and isolability. We assume that all parameters for the nominal dynamics (\ref{sub_dynamic}) are known, including the weights $\{w_{ij}\}$. Our conditions are in terms of the lumped state-space parameters {(\ref{lump_ss})} and the corresponding parameter perturbations. In other words, our results can be seen as conditions for either networks of single-integrators, or state-space modeled plants where the parameter perturbations do not necessarily result from topology failures.

\begin{definition} \label{defdistinguish} Consider $(\Phi, Q)$, $(\bar \Phi, Q)$ in (\ref{lump_ss}) and (\ref{lumpfaulty}) respectively. Let $y(x_0,\Phi,t)$ and $y(x_0,\bar \Phi, t)$ be the output signals of system (\ref{lump_ss}) and system (\ref{lumpfaulty}), respectively, with initial state $x_0$. We say $(\Phi, Q)$ and $(\bar \Phi, Q)$ are distinguishable (also say $\Phi$ and $\bar \Phi$ are distinguishable if $Q$ is implicitly known), if there exists $x_0\in {\mathbb R}^{n_x}$, such that $y(x_0,\Phi,t)-y(x_0,\bar \Phi,t)\not\equiv 0$, $t\ge 0$.
\end{definition}

{ As mentioned above,  Definition \ref{defdistinguish} does not need to hold for every initial state $x_0\in {\mathbb R}^{n_x}$. However, as will be shown in Proposition \ref{equal_isolation}, if $(\Phi, Q)$ and $(\bar \Phi, Q)$ are distinguishable, then almost all initial states {\emph{except a set of zero Lebesgue measure}} in ${\mathbb R}^{n_x}$ satisfy the inequality in Definition \ref{defdistinguish}.  We refer readers to \cite{Battistelli2015DetectingTV,BattistelliDetection2017,Patil2019IndiscernibleTV} for some characterizations of the initial states violating that inequality. It is also worthy to note that a related notion named output distinguishability can be found in \cite{Baglietto2014DistinguishabilityOD}, which requires that the corresponding system outputs (even from the same nominal system in two experiments) with not necessarily the same initial states should be different. Such definition is stricter than Definition \ref{defdistinguish} on $(\Phi, Q)$ and $(\bar \Phi, Q)$, and is often used for the offline scenario where the corresponding outputs may come from multiple experiments/processes or the initial states are unavailable.}

By Definitions \ref{defdection} and \ref{defdistinguish}, for networked system (\ref{sub_dynamic}) the link failure ${\cal E}_f$ is detectable, if and only if $(\Phi, Q)$ and $(\bar \Phi, Q)$ are distinguishable.  The following theorem gives necessary and sufficient conditions for distinguishability of $(\Phi,Q)$ and $(\bar \Phi, Q)$.
\begin{theorem} \label{algebraic_theorme1}
Given $(\Phi,Q)$ and $(\bar \Phi, Q)$ in (\ref{lump_ss}) and (\ref{lumpfaulty}) respectively, let the perturbation matrix $\Delta \Phi\doteq\Phi-\bar \Phi \in {\mathbb R}^{n_x\times n_x}$.  The following statements are equivalent:

(1) $(\Phi, Q)$ and $(\bar \Phi, Q)$ are distinguishable;

(2) $\left[
      \begin{array}{c}
        Q\Delta \Phi \\
        Q\Phi \Delta\Phi \\
        \vdots \\
        Q\Phi^{n_x-1}\Delta\Phi\\
      \end{array}
    \right]\ne 0$;

(3) The transfer function  $Q(\lambda I-\Phi)^{-1}\Delta\Phi\not\equiv 0$.
\end{theorem}

\begin{proof}(1) $\Leftrightarrow$ (2): For a given $x_0\in {\mathbb R}^{n_x}$, $y(x_0,{\cal G},t)=Qe^{\Phi t}x_0$, $y(x_0,\bar {\cal G}, t)=Qe^{\bar \Phi t}x_0$.
To make $Qe^{\Phi t}x_0-Qe^{\bar \Phi t}x_0\equiv 0$ for arbitrary $x_0\in {\mathbb R}^{n_x}$, $Qe^{\Phi t}-Qe^{\bar \Phi t}\equiv 0$ must hold. 
Notice that,
$$ \begin{aligned}  Qe^{\Phi t}&=Q(It+\Phi t+ \frac{1}{2}\Phi^2 t^2+\frac{1}{6}\Phi^3t^3+\cdots)\\
Qe^{\bar \Phi t}&=Q(It+\bar \Phi t+ \frac{1}{2}\bar \Phi^2 t^2+\frac{1}{6}\bar \Phi^3t^3+\cdots).\end{aligned}
$$Hence, $$Qe^{\Phi t}-Qe^{\bar \Phi t}=Q(\Phi-\bar \Phi)t+\frac{1}{2}Q(\Phi^2-\bar \Phi^2)t^2+\cdots.$$
Therefore, $Qe^{\Phi t}-Qe^{\bar \Phi t}\equiv 0$ requires that $Q(\Phi^i-\bar \Phi^i)=0$, for $i=1,...,\infty$. Notice that, if $Q(\Phi^{i-1}-\bar \Phi^{i-1})=0$ for some $i\ge 1$ (in fact it holds for $i=1$), then
$Q\Phi^i-Q\bar \Phi^{i}=Q\Phi^i-Q\bar \Phi^{i-1}\bar \Phi=Q\Phi^i-Q\Phi^{i-1}\bar \Phi=Q\Phi^{i-1}(\Phi-\bar \Phi)=Q\Phi^{i-1}\Delta\Phi.$  This means that, the condition $Q(\Phi^i-\bar \Phi^i)=0$ for $i=1,...,\infty$ is equivalent to $Q\Phi^{i-1}\Delta\Phi=0$ for $i=1,...,\infty$. According to the Cayley-Hamiltion theorem \cite{K.J.1988Multivariable}, if $Q\Phi^{i-1}\Delta\Phi=0$ for $i=1,...,n_x$, then for any $i\ge n_x+1$, there exists $(a_0,\cdots,a_{n_x-1})\in {\mathbb R}^{n_x}$, such that
$Q\Phi^{i-1}\Delta\Phi=\sum \nolimits_{i=1}^{n_x}a_{i-1}Q\Phi^{i-1}\Delta\Phi=0$. Hence, this proves that, (1) and (2) are equivalent.


(2) $\Leftrightarrow$ (3): We will first show that (2) $\Rightarrow$ (3), equivalently, if $Q(\lambda I-\Phi)^{-1}\Delta \Phi\equiv 0$, then $Q\Phi^{i-1}\Delta\Phi=0$ for $i=1,...,n_x$. In fact, when $\lambda>\rho(\Phi)$, it holds that
$$\begin{aligned} Q(\lambda I-\Phi)^{-1}\Delta\Phi &= Q\lambda^{-1}I(I+\lambda^{-1}\Phi+\lambda^{-2}\Phi^2+\cdots)\Delta \Phi \\
&=\sum\nolimits_{i=1}^{\infty}\lambda^{-i}Q\Phi^{i-1}\Delta\Phi.
\end{aligned}$$
To make $Q(\lambda I-\Phi)^{-1}\Delta \Phi=0$, each coefficient of $\lambda^{-i}$ must be zero. That is, $Q\Phi^{i-1}\Delta\Phi=0$ for $i=1,...,\infty$, which is equivalent to that $Q\Phi^{i-1}\Delta\Phi=0$ for $i=1,...,n_x$.

We are now proving (3) $\Rightarrow$ (2). Consider the converse-negative direction. Suppose that (2) is not true.  If $\lambda>\rho(\Phi)$, by the Cayley-Hamiltion theorem, there exists $(a_0,\cdots,a_{n_x-1})\in {\mathbb R}^{n_x}$, such that
$$
(\lambda I-\Phi)^{-1}=\lambda^{-1}\sum\nolimits_{i=0}^{\infty}(\lambda^{-1}\Phi)^{i}=\sum \nolimits_{i=0}^{n_x-1} \lambda^{-1-i}a_i\Phi^{i}.
$$
Hence, $Q(\lambda I-\Phi)^{-1}\Delta \Phi= \sum \nolimits_{i=0}^{n_x-1} \lambda^{-1-i}a_iQ\Phi^{i}\Delta\Phi=0$ holds for all $\lambda> \rho(\Phi)$. This further means that $Q(\lambda I-\Phi)^{-1}\Delta \Phi= 0$ for all $\lambda \in {\mathbb C}$. Hence we have (3) $\Rightarrow$ (2), which finishes the proof.
\end{proof}

Condition (3) of Theorem \ref{algebraic_theorme1} suggests the distinguishability of $(\Phi, Q)$ and $(\bar \Phi, Q)$ requires that, the perturbation $\Delta \Phi$ in the system state transition matrices can be inflected in the system output response. { From the derivations of Theorem \ref{algebraic_theorme1}, this condition does not depend on the observation time. In fact, if there exists $\tau\in (0,t_1]$ for some $t_1<\infty$ such that $Qe^{\Phi \tau}x_0-Qe^{\bar \Phi \tau}x_0\ne 0$ for a given initial state $x_0$ as in Definition \ref{defdistinguish}, then for arbitrary $t_2$ and $t_3$ satisfying $0\le t_2<t_3<\infty$, there exists a $\tau \in (t_2,t_3]$ ($\tau$ can be called the observation time) making $Qe^{\Phi \tau}x_0-Qe^{\bar \Phi \tau}x_0\ne 0$ (note that at least one entry of $Qe^{\Phi \tau}x_0-Qe^{\bar \Phi \tau}x_0$ is a non-identically zero polynomial of $\tau$ in this case).}

Consider the failure set ${\mathbb E}=\{{\cal E}_1,...,{\cal E}_r\}$. Let $\Phi_i$ be the lumped state transition matrix of the networked system after the link failure ${\cal E}_i$, $1\le i \le r$, which is defined in the same way as $\bar \Phi$ for ${\cal E}_f$, and let $\Phi_0\doteq \Phi$. From Definitions \ref{def1isolation} and \ref{defdistinguish}, ${\mathbb E}$ is isolable, if and only if for any two integers $i,j\in \{0,...,r\}$ with $i\ne j$, $(\Phi_i,Q)$ and $(\Phi_j,Q)$ are distinguishable. Combined with Theorem \ref{algebraic_theorme1}, this immediately leads to the following proposition.

\begin{proposition}\label{algebraic_isolation} For networked system (\ref{sub_dynamic}), a failure set ${\mathbb E}=\{{\cal E}_1,...,{\cal E}_r\}$ is isolable, if for any two integers $i,j\in \{0,...,r\}$, $i\ne j$, $Q(\lambda I- \Phi_i)^{-1}\Delta \Phi_{ij}\not\equiv 0$ holds, where $\Delta \Phi_{ij}\doteq \Phi_i-\Phi_j$.
\end{proposition}


From their derivations, Theorem \ref{algebraic_theorme1} and Proposition \ref{algebraic_isolation} are valid for arbitrary parameter perturbations $\Delta \Phi$ (or $\Delta \Phi_{ij}$) not necessarily resulting from topology failures. On the basis of Proposition \ref{algebraic_isolation}, we give a property of an isolable failure set as follows.
\begin{proposition} \label{equal_isolation} For networked system (\ref{sub_dynamic}), if a failure set ${\mathbb E}=\{{\cal E}_1,...,{\cal E}_r\}$ is isolable, then there exists a common $x_0\in {\mathbb R}^{n_x}$, such that for any $i,j\in \{0,....,r\}$, ~$i\ne j$, $y(x_0,{\bar G}_i,t)-y(x_0,{\bar G}_j,t)\not\equiv 0$ holds. Moreover, denote the set of all $x_0$ satisfying the aforementioned condition by $X_0\subseteq {\mathbb R}^{n_x}$. Then, ${\mathbb R}^{n_x}\backslash X_0$ has Lebesgue measure zero  in ${\mathbb R}^{n_x}$.
\end{proposition}
\begin{proof} Notice that
 $$\begin{aligned}y(x_0,{\cal G}_i,t)-y(x_0,{\cal G}_j,t)&=Qe^{\Phi_it}x_0-Qe^{\Phi_jt}x_0\\
 &=[Q,-Q]e^{{\bf diag}\{\Phi_i,\Phi_j\}t}\left[
                                                                                         \begin{array}{c}
                                                                                           I \\
                                                                                           I \\
                                                                                         \end{array}
                                                                                       \right]x_0.\end{aligned}
 $$
 Following the proof of Theorem \ref{algebraic_theorme1}, substitute the Taylor expansion of $e^{{\bf diag}\{\Phi_i,\Phi_j\}t}$ into the above formula, use the Cayley-Hamiltion theorem,  and we obtain that $Qe^{\Phi_it}x_0-Qe^{\Phi_jt}x_0\ne 0$, if and only if
 $$\underbrace{{\bf col}\left\{[Q,-Q]\left[
                                                                                  \begin{array}{cc}
                                                                                    \Phi_i^k & 0 \\
                                                                                    0 & \Phi_j^k \\
                                                                                  \end{array}
                                                                                \right]\right.\left.\left[
                                                                                         \begin{array}{c}
                                                                                           I \\
                                                                                           I \\
                                                                                         \end{array}
                                                                                       \right]\Big|_{k=1}^{2n_x-1}\right\}}_{\doteq F_{ij}}x_0 \ne 0.$$
 Hence, if $x_0\notin {\bf ker}(F_{ij})$, then $Qe^{\Phi_it}x_0-Qe^{\Phi_jt}x_0\ne 0$. If $\mathbb E$ is isolable, by Proposition \ref{algebraic_isolation}, $F_{ij}\ne 0$. Hence, ${\bf ker}(F_{ij})$ is a proper subspace of ${\mathbb R}^{n_x}$. In addition, $\bigcup \nolimits_{0\le i < j\le r} {\bf ker}(F_{ij})$ is also a proper subspace of ${\mathbb R}^{n_x}$ and has Lebesgue measure zero  in ${\mathbb R}^{n_x}$, since the union of any finite number of proper subspaces of ${\mathbb R}^{n_x}$ is a proper subspace of ${\mathbb R}^{n_x}$. Therefore, any $x_0$ in ${\mathbb R}^{n_x}\backslash {\bigcup \nolimits_{0\le i < j\le r} {\bf ker}(F_{ij})}$ makes $y(x_0,{\bar G}_i,t)-y(x_0,{\bar G}_j,t)\not\equiv 0$, for $i,j\in \{0,....,r\}$, $i\ne j$.
\end{proof}

{ Proposition \ref{equal_isolation} indicates that, a randomly generated initial state $x_0\in {\mathbb R}^{n_x}$ almost surely results in output responses that can isolate the exact failure from an isolable failure set $\mathbb E$.  Concerning the isolation implementation, with the knowledge of the faultless dynamics (\ref{sub_dynamic}) and the prior failure set $\mathbb E$ for a randomly generated initial state $x_0$, one {\emph{possible (centralized)}} approach may be using a bank of least square estimators (c.f. \cite{Battistelli2015DetectingTV}) or observer-based residual generators \cite{Ding2008ModelbasedFD} to distinguish every two of the candidate failures. Some data-driven approaches might also be possible candidates \cite{Zhou2020ReviewOD}. This is left for the future research.} 



\section{Genericity of Failure Detectability and Isolability}
From now on, we deal with the situation where the exact values of $\{w_{ij}\}$ are unknown, but their zero-nonzero patterns are accessible. We call a set of real values for $\{w_{ij}\}$ with the corresponding zero-nonzero patterns a weight realization. A property  is called {\emph{generic}}, if either for almost all weight realizations of $\{w_{ij}\}$ except for a set with Lebesgue measure zero  in the corresponding parameter space,  this property holds true, or for all weight realizations of $\{w_{ij}\}$, this property does not hold. In this section, we will prove that, failure detectability and isolability are generic properties for the considered networked systems. 

\begin{proposition}\label{generic_detect} For networked system (\ref{sub_dynamic}) with known $(A,H,C)$ and zero-nonzero patterns of the weights $\{w_{ij}\}$, detectability of a failure ${\cal E}_f\subseteq {\cal E}$ is a generic property.
\end{proposition}

\begin{proof} Let $z_1,...,z_{|\cal E|}$ be free parameters in $\{w_{ij}\}$ that can take nonzero real values independently. Assume that ${\cal E}_f$ is undetectable, which requires that $Q\Phi^{k-1}\Delta \Phi=0$ for $k=1,...,n_x$, by Theorem \ref{algebraic_theorme1}. Each $Q\Phi^{k-1}\Delta \Phi=0$ induces at most $n_x^2$ scalar equations, and assume that through $k=1,...,n_x$, there are in total $q$ informative constraints (meaning that none of these constraints is a linear combination of the rest), denoted by
$$\left\{\begin{array}{l} f_1(z_1,...,z_{|\cal E|})=0\\
\quad \quad \vdots \\
f_q(z_1,...,z_{|\cal E|})=0, \end{array}\right.
$$
where each $f_i(z_1,...,z_{|\cal E|})$ is a polynomial of $(z_1,...,z_{|\cal E|})$ with real coefficients. These constraints are equivalent to $F(z_1,...,z_{|\cal E|})\doteq \sum\nolimits_{i=1}^qf_i^2(z_1,...,z_{|\cal E|})=0$. As $F(z_1,...,z_{|\cal E|})$ is a polynomial of $(z_1,...,z_{|\cal E|})$, if it is not identically zero, then for almost all values of $(z_1,...,z_{|\cal E|})$ except for the proper algebraic variety $\{(z_1,...,z_{|\cal E|})\in {\mathbb R}^{|\cal E|}:F(z_1,...,z_{|\cal E|})=0\}$ with Lebesgue measure zero  in ${\mathbb R}^{|\cal E|}$, $F(z_1,...,z_{|\cal E|})\ne 0$; otherwise, for all values of $(z_1,...,z_{|\cal E|})$ in ${\mathbb R}^{|\cal E|}$, $F(z_1,...,z_{|\cal E|})= 0$. This proves the proposed statement.
\end{proof}
Note that $\Phi_i$ and $\Phi_j$ are both obtained from $\Phi$ by zeroing entries of $W$ corresponding to ${\cal E}_i$ and ${\cal E}_j$ respectively. An immediate result from  Propositions \ref{algebraic_isolation} and \ref{generic_detect} is that, distinguishability of $(\Phi_i,Q)$ and $(\Phi_j,Q)$ is a generic property for networked system (\ref{sub_dynamic}), $i,j\in \{0,...,r\}$.
\begin{proposition}\label{generic_isolation} For networked system (\ref{sub_dynamic}),  isolability of a failure set ${\mathbb E}=\{{\cal E}_1,...,{\cal E}_r\}$ is a generic property.
\end{proposition}
\begin{proof}By Proposition \ref{algebraic_isolation}, the statement follows from Proposition \ref{generic_detect} and the fact that the union of a finite number of proper algebraic varieties in ${\mathbb R}^{|\cal E|}$ also has Lebesgue measure zero  in ${\mathbb R}^{|\cal E|}$.
\end{proof}

\begin{example}[Genericity of Detectability and Isolability] Consider a networked system of single-integrators. Let $$\Phi=\left[
                                                                                                                            \begin{array}{cccc}
                                                                                                                              0 & 0 & 0 & a_1 \\
                                                                                                                              a_2 & 0 & a_3 & 0 \\
                                                                                                                              0 & 0 & 0 & a_4 \\
                                                                                                                              a_5 & 0 & 0 & 0 \\
                                                                                                                            \end{array}
                                                                                                                          \right],Q=[0, 0, 1, 0].$$Denote $Z=(a_1,...,a_5)$, and $\Phi_0\doteq \Phi$. Consider two failures ${\cal E}_1=\{(1,4)\}$ and ${\cal E}_2=\{(4,3)\}$.  We obtain $$Q(\lambda I- \Phi_0)^{-1}\Delta \Phi_{01}=[ -\frac{a_5(a_1a_2 + a_3a_4)}{-\lambda^3+ a_1a_5\lambda}, 0, 0, 0],$$ $$Q(\lambda I- \Phi_0)^{-1}\Delta \Phi_{02}=Q(\lambda I- \Phi_1)^{-1}\Delta \Phi_{12}=[0, 0, 0, -a_3a_4].$$ Hence, ${\cal E}_1$ is detectable in the set $\{Z\in {\mathbb R}^5: a_5(a_1a_2 + a_3a_4)\ne 0\}$. And $\{{\cal E}_1, {\cal E}_2\}$ is isolable in $\{Z\in {\mathbb R}^5: a_5(a_1a_2 + a_3a_4)\ne 0, a_3a_4\ne 0\}$. The complements of both sets are of zero Lebesgue measure in ${\mathbb R}^5$. \hfill $\square$

\end{example}

The above two propositions reveal that, it is the topology of the faultless networked system, rather than the exact weights of the subsystem links, that dominates detectability and isolability of a given failure (set). We say that a failure ${\cal E}_f$ is {\emph{generically detectable}}, if for almost all weight realizations of $\{w_{ij}\}$, ${\cal E}_f$ is detectable for the corresponding networked systems. Similarly, a failure set $\mathbb E$ is {\emph{generically isolable}}, if for almost all weight realizations of $\{w_{ij}\}$, $\mathbb E$ is isolable for the corresponding networked systems.
From Propositions \ref{generic_detect} and \ref{generic_isolation}, if there exists one weight realization for $\{w_{ij}\}$ such that a given failure is detectable for the corresponding system, then this failure is generically detectable for the networked systems. Such property holds true for generic isolability. 

%

\section{Graph-theoretic Conditions for Generic Detectability and Isolability} \label{fengwei}
In this section, graph-theoretic conditions for generic detectability and isolability of a failure (set) are given for the networked systems. 
\subsection{Conditions for Generic Detectability} \label{section_detection}
To present the conditions for generic detectability, we first introduce some definitions. For a failure ${\cal E}_f\subseteq {\cal E}$, let $V_{R}({\cal E}_f)$ denote the set of ending nodes of ${\cal E}_f$. Recall that $\cal S$ is the set of locations of sensors.  Define a distance index $d_{\min}$ of $\cal G$ as
$$d_{\min}=\min \limits_{v\in {V_{R}({\cal E}_f)}, u\in {\cal S}} {\rm dist}(v,u,{\cal G}). $$ That is, $d_{\min}$ is the shortest distance from the ending nodes of ${\cal E}_f$ to nodes that are directly measured (i.e., sensor nodes). For each subsystem, define a transfer function $H_s(\lambda)\doteq (\lambda I- A)^{-1}H$. Define a transfer index $r_{\max}$ for subsystems as
$$r_{\max}=\left\{ \begin{aligned} i, CH_s^i(\lambda)&\ne 0, CH_s^{i+1}(\lambda)=0, i\in {\mathbb N}, \\
\infty, CH_s^i(\lambda)&\ne 0, \forall i\in {\mathbb N}.\end{aligned}\right. $$
That is, $r_{\max}$ is the maximum exponent $i$ such that~$CH^i_s(\lambda)\ne 0$. To give conditions for generic detectability, we need the following {{intermediate}} results.
\begin{lemma}[\cite{K.J.1988Multivariable}] \label{distance} Let $M$ be an adjacency matrix of a digraph $\cal D$ with node set $\{1,...,N\}$. Then, i) $[M^k]_{ij}= 0$ if $k< {\rm dist}(j,i,{\cal D})$; ii) $[M^k]_{ij}\ne 0$ only if there is path from $j$ to $i$ with length $k$.
\end{lemma}

\begin{lemma}\label{subsequence} Given $A,H\in {\mathbb R}^{n\times n}$, let $H_s(\lambda)=(\lambda I-A)^{-1}H$. Let $\{n_i\}_{i=1}^{i_{\max}}$ be any (infinite or finite) subsequence of $\{1,2,\cdots,\infty\}$. Then, there exists a dense set $\bar {\mathbf \Lambda}\subseteq {\mathbb C}$, such that when $\lambda \in \bar{\mathbf \Lambda}$, $I+\sum\nolimits_{i=1}^{i_{\max}}H^{n_i}_s(\lambda)$ is invertible.
\end{lemma}
\begin{proof} Let $\{\lambda_k\}_{k=1}^n$ be the eigenvalues of $H_s(\lambda)$. Then, the eigenvalues of $I+\sum\nolimits_{i=1}^{i_{\max}}H^{n_i}_s(\lambda)$ are $\{1+\sum\nolimits_{i=1}^{i_{\max}}\lambda_k^{n_i}\}_{k=1}^n$. Hence, there exists some dense set $\bar {\mathbf \Lambda}$ such that $\rho(H_s(\lambda))$ is small enough if $\lambda \in \bar {\mathbf \Lambda}$,{\footnote{Let $\rho_{\min}(\cdot)$, $\sigma_{\min}(\cdot)$ and $\sigma_{\max}(\cdot)$ denote the minimum eigenvalue, minimum and maximum singular values, respectively. We have $\rho(H_s(\lambda))\le \sigma_{\max}(H_s(\lambda))\le \sigma_{\max}((\lambda I-A)^{-1})\sigma_{\max}(H)= \sigma^{-1}_{\min}(\lambda I-A)\sigma_{\max}(H)$. Note that $\sigma_{\min}(\lambda I-A)=\rho_{\min}^{\frac{1}{2}}((\lambda I-A)(\lambda I-A)^{\intercal})=\rho_{\min}^{\frac{1}{2}}(\lambda^2I-\lambda (A+A^{\intercal})+AA^{\intercal})\ge \rho^{\frac{1}{2}}_{\min}(\lambda^2I-\lambda (A+A^{\intercal}))\ge(\lambda^2- \lambda \rho(A+A^{\intercal}))^{\frac{1}{2}}$ when $\lambda\ge \frac{1}{2}\rho(A+A^{\intercal})$. Hence, when $\lambda$ is large enough, $\rho (H_s(\lambda))$ is small enough.}}  making ${\bf abs}(\sum\nolimits_{i=1}^{i_{\max}}\lambda_k^{n_i})\le \sum\nolimits_{i=1}^{i_{\max}} \rho(H_s(\lambda))^{n_i}<1$. Consequently, all eigenvalues of $I+\sum\nolimits_{i=1}^{i_{\max}}H^{n_i}_s(\lambda)$ are nonzero.
\end{proof}

\begin{theorem} \label{main_generic_dect} For networked system (\ref{sub_dynamic}) with known $(A,H,C)$ and zero-nonzero patterns of the weights $\{w_{ij}\}$, a failure ${\cal E}_f\subseteq {\cal E}$ is generically detectable, if and only if {\footnote{If $r_{\max}=\infty$ and $d_{\min}=\infty$, this inequality does not hold.}}
\begin{equation}\label{graphcondition} d_{\min}\le r_{\max}-1.\end{equation}
\end{theorem}

\begin{proof} Note that Condition (3) of Theorem \ref{sub_dynamic} can be used to prove this theorem. We first derive a formula which is used for proving both necessity and sufficiency. Let $\bar W$ be the adjacency matrix of $({\cal V}, {\cal E}\backslash {\cal E}_f)$. Recall that $\bar \Phi=I_N\otimes A+ \bar W \otimes H$. Define $\Delta W\doteq W- \bar W$. Then, $\Delta \Phi = \Phi-\bar \Phi =\Delta W \otimes H$. The corresponding transfer function becomes
\[\begin{aligned}
&G_f(\lambda)\doteq Q(\lambda I-\Phi)^{-1}\Delta \Phi \\
&=S\otimes C(\lambda I_{n_x}- I_N\otimes A- W\otimes H)^{-1}\Delta W\otimes H.\end{aligned}\]{Noting that $(\lambda I_{n_x}- I_N\otimes A- W\otimes H)=I_N\otimes(I-A)[I_{n_x}-I_N\otimes(I-A)^{-1}W\otimes H]$, we have\begin{equation}\begin{aligned}
G_f(\lambda)&=S\otimes C[I_{n_x}-W\otimes(\lambda I-A)^{-1} H]^{-1}I_N\otimes (\lambda I-A)^{-1}\\ &\cdots \Delta W\otimes H.\\
\end{aligned}\end{equation}}When $\lambda\in {\mathbf{\Lambda}}\doteq \{\lambda \in {\mathbb  C}: \rho(W)\rho((\lambda I-A)^{-1} H)<1\}$, which is dense in $\mathbb C$, rewrite $G_f(\lambda)$ as
\begin{equation}\label{basic_exp2} \begin{aligned}
&G_f(\lambda)=S\otimes C\sum \limits_{k=0}^{\infty}[W\otimes(\lambda I-A)^{-1} H]^{k}\Delta W \otimes(\lambda I-A)^{-1}H\\
&=\sum \limits_{k=0}^{r_{\max}} (S\otimes C)W^k\otimes [(\lambda I-A)^{-1} H]^{k}\Delta W \otimes(\lambda I-A)^{-1}H\\
&= \sum \limits_{k=0}^{r_{\max}-1} SW^k\Delta W\otimes C[(\lambda I-A)^{-1} H]^{k+1}.\\
\end{aligned}
\end{equation}

{\emph{Necessity:}} Suppose that (\ref{graphcondition}) is not true. Then, either i) $d_{\min}=0$, $r_{\max}=0$ or ii) $d_{\min}\ge 1$, and $d_{\min}>r_{\max}-1$.
In case i), as $r_{\max}=0$, we have $C[(\lambda I-A)^{-1} H]^k=0$ for $k\ge 1$. Hence, $G_f(\lambda)=0$ for $\lambda \in {\mathbf\Lambda}$, which means $G_f(\lambda)\equiv 0$, leading to the undetectability of ${\cal E}_f$. In case ii), without losing generality, suppose that the sensor nodes are indexed  as $1,...,|\cal S|$, and the ending nodes of failure ${\cal E}_f$ as $q+1,...,N$, $q\ge |\cal S|$. Then, if $k<d_{\min}$, we have the following partitions:
\begin{equation}\label{partition} S=[I_{|{\cal E}|},0,0], W^k\!\!=\!\! \left[
                                \begin{array}{ccc}
                                  W^k_{11} & W^k_{12} & 0 \\
                                  W^k_{21} & W^k_{22} & W^k_{23} \\
                                  W^k_{31} & W^k_{32} & W^k_{33} \\
                                \end{array}
                              \right], \Delta W\!\!=\!\!\left[\!
                                                  \begin{array}{c}
                                                    0 \\
                                                    0 \\
                                                    \Delta W_3 \\
                                                  \end{array}
                                                \!\right]
\end{equation}where $W^k_{11}$, $W^k_{22}$, $W^k_{33}$, and $\Delta W_3$ have dimensions respectively $|{\cal S}|\times |{\cal S}|$, $(q-|{\cal S}|)\times (q-|{\cal S}|)$, $(N-q)\times (N-q)$, and $(N-q)\times N$, and the rest have compatible dimensions. Note that the $(1,3)$th block of $W^k$ is zero due to Lemma \ref{distance} and the fact that $k<d_{\min}$. It is easy to see that
\begin{equation}\label{necessity}SW^k\Delta W=0, \forall k\in \{0,...,d_{\min}\}.\end{equation}
Hence, $G_f(\lambda)=0$ for $\lambda \in {\mathbf \Lambda}$ from (\ref{basic_exp2}), making $G_f(\lambda)\equiv 0$. Thus, ${\cal E}_f$ is always undetectable.

Sufficiency: By genericity of detectability, to show sufficiency it is enough to construct a weight realization $\{w_{ij}\}$ associated with which ${\cal E}_f$ is detectable. By reordering nodes, suppose that the {\emph{shortest}} path from $V_R({\cal E}_f)$ to $\cal S$ in $\cal G$ is ${\cal P}\doteq \{(\bar d,d_{\min}), (d_{\min}-1,d_{\min}-2),\cdots,(2,1)\}$, and $(i^*,\bar d)\in {\cal E}_f$, where $\bar d\doteq d_{\min}+1$. Let the weights of links in ${\cal E}\backslash (\{(i^*,\bar d)\} \cup {\cal P})$ be zero (then links in ${\cal E}_f\backslash \{(i^*,\bar d)\}$ have zero weights), whiles links in $\{(i^*,\bar d)\} \cup {\cal P}$ have weight $1$. Then,
\begin{equation}\label{construction}W=e^{[N]}_{\bar d,i^*}+\sum \limits_{i=1}^{d_{\min}} e^{[N]}_{i,i+1}, \Delta W= e^{[N]}_{\bar d,i^*}.\end{equation} We consider two cases. See Fig. \ref{example_proof}.

Case 1), $i^*>d_{\min}+1$ (Fig. \ref{subeproofa}). Without losing generality, let $i^*=d_{\min}+2$. Since each nonzero entry of $W$ is $1$, from Lemma \ref{distance}, $[W^{k}]_{1\bar d}=0$ for $k\in\{0,1,\cdots\}\backslash \{d_{\min}\}$, and $[W^{d_{\min}}]_{1\bar d}=1$.
 Hence, considering ${\cal S}=\{1\}$, we have $SW^k\Delta W= [e^{[N]}_1]^{\intercal}W^ke^{[N]}_{\bar d,i^*}=0$ if $k\in\{0,1,\cdots,\infty\}\backslash \{d_{\min}\}$, and $SW^k\Delta W=[e^{[N]}_{i^*}]^{\intercal }$ if $k=d_{\min}$.
 Consequently, $G_f(\lambda)=[e^{[N]}_{i^*}]^{\intercal}\otimes CH_s(\lambda)^{d_{\min}+1}\ne 0$ from (\ref{basic_exp2}), making ${\cal E}_f$ detectable.

 Case 2), $i^*\in \{1,...,d_{\min}+1\}$ (Fig. \ref{subeproofb}).  In this case, from Lemma \ref{distance}, we have
 $$[W^{k}]_{1\bar d}=\left\{\begin{array}{l} 1, {\text{ if there is a path from $\bar d$ to $1$ with length}} \ k  \\ 0, {\text{otherwise}}\end{array} \right.$$
Suppose $[W^{k}]_{1\bar d}=1$ for $k\in \{n_1,n_2,...,\infty\}$, where $n_1=d_{\min}$. Considering ${\cal S}=\{1\}$,  we have $SW^k\Delta W=[e^{[N]}_1]^{\intercal}W^ke^{[N]}_{\bar d,i^*}=[e^{[N]}_{i^*}]^{\intercal}$ for $k\in \{n_1,n_2,...,\infty\}$, and otherwise $SW^k\Delta W=0$. Substituting these into (\ref{basic_exp2}), we get
$$
\begin{aligned} G_f(\lambda)&= [e^{[N]}_{i^*}]^{\intercal}\otimes \left\{C \sum \limits_{i=1}^{k_{\max}} [(\lambda I-A)^{-1} H]^{n_i+1}\right\}\\
&=[e^{[N]}_{i^*}]^{\intercal}\otimes CH^{d_{\min}+1}_s(\lambda)\left\{I+\sum \limits_{i=2}^{k_{\max}}H^{n_i-n_1}_s(\lambda)\right\},
\end{aligned}$$where $k_{\max}\doteq \max \{k: n_k\le r_{\max}-1\}$. As (\ref{graphcondition}) holds, the value of $\lambda$ making $CH^{d_{\min}+1}_s(\lambda)\ne 0$ is everywhere dense in $\mathbb C$. Together with Lemma \ref{subsequence}, we know there exists a dense set $\hat {\mathbf \Lambda} \subseteq \mathbb C$, such that for $\lambda \in \hat {\mathbf \Lambda}$,   $CH^{d_{\min}+1}_s(\lambda)\ne 0$ and $I+\sum \nolimits_{i=2}^{k_{\max}}H^{n_i-n_1}_s(\lambda)$ is invertible, making  $G_f(\lambda)\ne 0$. This proves the detectability of ${\cal E}_f$ by Theorem \ref{algebraic_theorme1}.
\end{proof}

Theorem \ref{main_generic_dect} gives a graph-theoretic condition for generic failure detectability.  Notice that $H_s(\lambda)$ is a transfer function from the internal input to the internal output of a subsystem. A deep insight of Theorem \ref{main_generic_dect} indicates that, the necessary and sufficient condition for generic detectability of failure ${\cal E}_f$ is that, at least one sensor should receive signals from at least one ending node of the faulty links. 

When $A=0\in{\mathbb R}^{1\times 1}, B=\Gamma=C=1$, system (\ref{sub_dynamic}) collapses to a networked system of single-integrators, or alternatively speaking, the conventional {\emph{structured system}} where every entry in the system matrices is either fixed zero or a free parameter \cite{generic}. In this case, $r_{\max}=\infty$. Theorem \ref{main_generic_dect} immediately leads to the following result.
\begin{corollary}[Generic detectability for structured system]\label{single_failure} For a networked system of single-integrators (or a structured system), a failure ${\cal E}_f$ is generically detectable, if and only if there exists a path from one ending node of ${\cal E}_f$ to one of the sensor nodes in $\cal G$.
\end{corollary}

\begin{example}\label{example2}
Consider a networked system with $5$ subsystems. The parameters for subsystem dynamics are respectively $$A=\left[
                                                                                                      \begin{array}{ccc}
                                                                                                        1 & -1 & 0 \\
                                                                                                        0 & 2 & 0 \\
                                                                                                        0 & 0 & -1 \\
                                                                                                      \end{array}
                                                                                                    \right]
,H=\left[
        \begin{array}{ccc}
          0 & 2 & 0 \\
          0 & 0 & 1 \\
          0 & 0 & 0 \\
        \end{array}
      \right]
,C=[1,0,0].$$The sensor is located at node $1$, i.e., ${\cal S}=\{1\}$.  The network topology is shown in Fig. \ref{example_5nodes}. For this networked system, $C[(\lambda I-A)^{-1}H]^2\ne 0$ whiles $C[(\lambda I-A)^{-1}H]^3= 0$. Hence, $r_{\max}=2$. For the link failure $\{(1,2)\}$, $d_{\min}=2> r_{\max}-1$. From Theorem \ref{main_generic_dect}, failure $\{(1,2)\}$ is undetectable irrespective of weights of these links. This can be validated by the algebraic conditions in Theorem \ref{algebraic_theorme1} when any exact weights are given.

On the other hand, for the failure $\{(2,5)\}$, $d_{\min}=1=r_{\max}-1$. From Theorem \ref{main_generic_dect}, this failure is generically detectable, which can be validated using Theorem \ref{algebraic_theorme1} on randomly generated weights. In fact, among the single-link failures, only the failure $\{(2,5)\}$, $\{(4,5)\}$, or $\{(5,1)\}$ can be detectable (see Fig. \ref{example_5nodes}). \hfill $\square$
\end{example}

\begin{figure}
  \centering
  \subfigure[]{ \label{subeproofa}
  \includegraphics[width=2.3in]{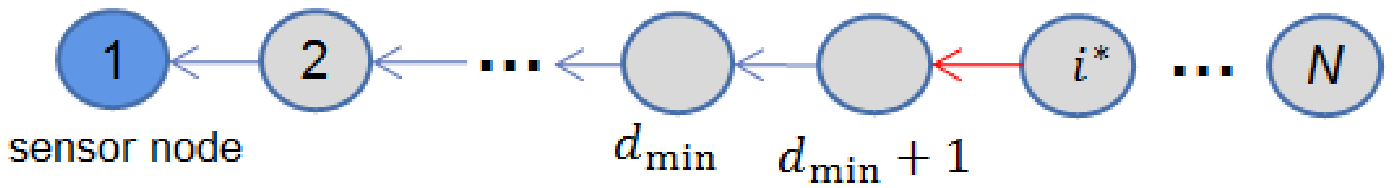}} \hspace{1in}
  \subfigure[]{\label{subeproofb}
   \includegraphics[width=2.3in]{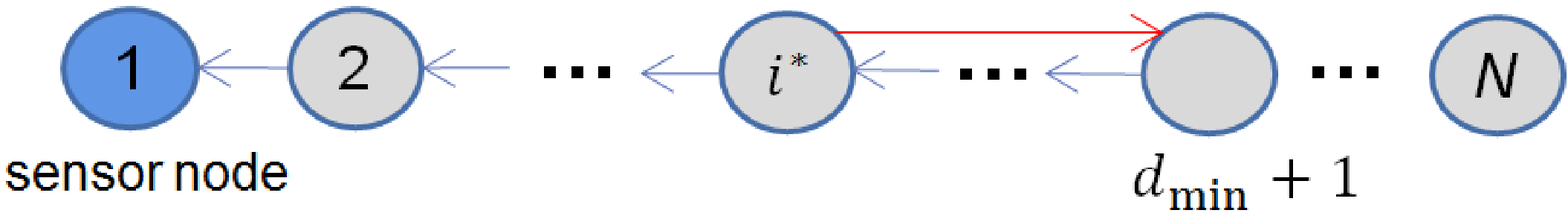}}
  \caption{Network topology in the proof of Theorem \ref{main_generic_dect}. (a): case 1). (b): case~2). Links in red are faulty.}\label{example_proof}
\end{figure}

\begin{figure}
  \centering
  \subfigure[]{ \label{subexpa}
  \includegraphics[width=1in]{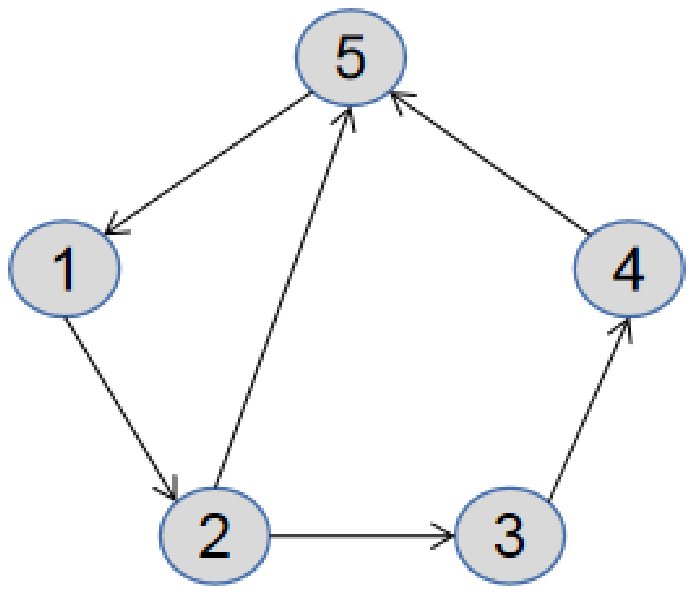}} \hspace{1in}
  \subfigure[]{\label{subexpb}
   \includegraphics[width=1in]{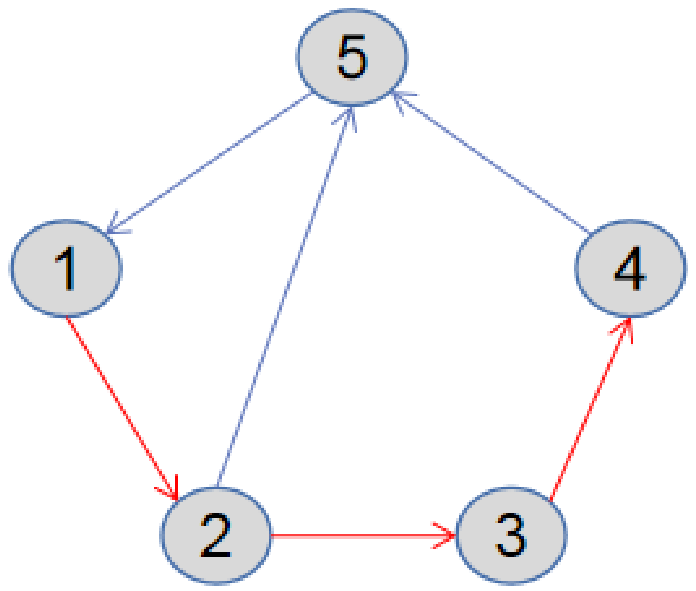}}
  \caption{Network topology in Example \ref{example2}. (a): original network topology. (b): links in blue are detectable, and in red are undetectable.}\label{example_5nodes}
\end{figure}

\subsection{Conditions for Generic Isolability}
Consider a prior failure set ${\mathbb E}=\{{\cal E}_1,...,{\cal E}_r\}$. For each $i\in \{1,...,r\}$, let $W_i$ be the adjacency matrix for ${\cal G}_i\doteq ({\cal V}, {\cal E}\backslash {\cal E}_i)$, which is defined in the same way as $\bar W$ for $\bar {\cal G}$. For $i,j\in \{0,...,r\}$ with $i\ne j$, define $\Delta W_{ij}=W_i-W_j$, and ${\cal E}_{ij}={\cal E}_{i}\cup {\cal E}_{j}\backslash ({\cal E}_i\cap {\cal E}_j)$. That is, ${\cal E}_{ij}$ is the link set of the digraph whose adjacency matrix is $\Delta W_{ij}$, reflecting the difference between ${\cal G}_i$ and ${\cal G}_j$. Moreover, define the distance index $d_{ij}$ as
$$d_{ij}=\min \limits_{v\in {V_{R}({\cal E}_{ij})}, u\in {\cal S}} {\rm dist}(v,u,{{\cal G}_i}). $${{It is easy to verify from the definition that $d_{ij}=d_{ji}$. }}Define $$d_{\min}^{\mathbb E}=\max\limits_{0\le i < j\le r}\{d_{ij}\}. $$

Before giving conditions for generic isolability, we present the condition for generic distinguishability of $\Phi_i$ and $\Phi_j$, recalling that they are lumped state transition matrices of the networked system after failures ${\cal E}_i$ and ${\cal E}_j$, respectively.
\begin{proposition}\label{generic_distinguish}  For networked system (\ref{sub_dynamic}), $\Phi_i$ and $\Phi_j$ are generically distinguishable, if and only if $d_{ij}\le r_{\max}-1$.
\end{proposition}
\begin{proof}
By regarding $W_i$ as $W$ and $\Delta W_{ij}$ as $\Delta W$, the proof follows similar arguments to that of Theorem \ref{main_generic_dect}. The only difference lies in that $\Delta W_{ij}$ may contain some nonzero entries which do not appear in $W_i$ (see (\ref{partition}) and (\ref{construction}) respectively). In the proof for necessity,  this difference does not violate (\ref{necessity}), as the corresponding partitions like (\ref{partition}) still hold. In the proof for sufficiency, such difference leads to that $W$ may not contain $e^{[N]}_{\bar d,i^*}$ in (\ref{construction}). It is an easy manner to validate such difference does not violate the validness of the remaining arguments.
\end{proof}


\begin{theorem} \label{main_generic_iso} Consider networked system (\ref{sub_dynamic}) with known $(A,H,C)$ and zero-nonzero patterns of weights $\{w_{ij}\}$.  A failure set ${\mathbb E}\!=\!\{{\cal E}_1,...,{\cal E}_r\}$ is generically isolable, if and only~if
\begin{equation}\label{graphcondition2} d_{\min}^{\mathbb E}\le r_{\max}-1,\end{equation}where the transfer index $r_{\max}$ is defined in Section \ref{section_detection}.
\end{theorem}
\begin{proof} This theorem is based on Propositions \ref{algebraic_isolation}, \ref{generic_detect}, \ref{generic_distinguish} and Theorem \ref{main_generic_dect}. For necessity, if (\ref{graphcondition2}) is not true,  then there exist two integers $i,j\in\{0,...,r\}$ such that $d_{ij}>r_{\max}-1$. From Proposition \ref{generic_distinguish}, $\Phi_i$ and $\Phi_j$ are not generically distinguishable, which means that $\mathbb E$ is not generically isolable.

For sufficiency, let $Z=(z_1,...,z_{|\cal E|})$ be free parameters in $\{w_{ij}\}$ that can take values independently. For each pair $i,j\in\{0,...,r\}$, $i<j$,
following Proposition \ref{generic_distinguish}, a numerical realization for $Z$ exists so that $\Phi_i$ and $\Phi_j$ are distinguishable. From Proposition \ref{generic_detect}, the set of values for $Z$ making $\Phi_i$ and $\Phi_j$ not distinguishable, denoted by ${\cal P}_{ij}$, has Lebesgue measure zero  in ${\mathbb R}^{|\cal E|}$. {As $\bigcup\nolimits_{0\le i< j\le r}{\cal P}_{ij}$ still has zero Lebesgue measure in  ${\mathbb R}^{|\cal E|}$}, there always exists $Z$ in ${\mathbb R}^{|\cal E|}\backslash (\bigcup\nolimits_{0\le i< j\le r}{\cal P}_{ij})$  making $\Phi_i$ and $\Phi_j$ distinguishable, for each pair $(i,j)$ with $0\le i< j\le r$. With Proposition \ref{algebraic_isolation}, this proves the sufficiency.
\end{proof}

\begin{remark}In Theorem \ref{main_generic_iso}, determining $d_{\min}^{\mathbb E}$ requires computing $d_{ij}$ for ${\tiny{\left(\begin{array}{l} r+1\\ \ \ \ 2 \end{array}\right)}}$ times, which grows quadratically with $|\mathbb E|$. When $|\mathbb E|$ grows exponentially with $|\cal E|$, this is still a huge computation cost. It is excepted that, exploring the inherent structures of $\mathbb E$ may sometimes avoid computing all $d_{ij}$ (c.f., Proposition \ref{equivalent}).\end{remark}

Theorems \ref{main_generic_dect} and \ref{main_generic_iso} give some fundamental structural limitations for the networked system to support detectability and isolability of a failure (set). These conditions must be satisfied before whatever detection and isolation algorithms are valid.

{ Using Theorems \ref{main_generic_dect} and \ref{main_generic_iso}, an interesting finding is that, if every two elements of a failure set $\mathbb E$ do not intersect (in terms of edges) then the generic isolability of $\mathbb E$ is equivalent to the generic detectability of every element of $\mathbb E$; see the following proposition.

\begin{proposition}[A class of generically isolable failure sets] \label{equivalent}
In networked system (\ref{sub_dynamic}), given a failure set ${\mathbb E}=\{{\cal E}_1,...,{\cal E}_r\}$,  if ${\cal E}_i\cap {\cal E}_j=\emptyset$ $\forall i\ne j$, then $\mathbb E$ is generically isolable, if and only if each ${\cal E}_i$ is generically detectable for $i=1,...,r$.
\end{proposition}

\begin{proof} The generic detectability of every element of $\mathbb E$ is obviously necessary for generic isolability of $\mathbb E$ by Definition \ref{def1isolation}.
Now suppose that every element of $\mathbb E$ is generically detectable. Consider arbitrarily ${\cal E}_i$ and ${\cal E}_j$ with $i\ne j$.  Since ${\cal E}_i\cap {\cal E}_j=\emptyset$, we have ${\cal E}_{ij}={\cal E}_i\cup {\cal E}_j$. As ${\cal E}_i$ is generically detectable, there exists a path ${\cal P}_i\doteq \{e_1,...,e_{k}\}$ in $\cal G$ such that the starting node of link $e_1$ belongs to $V_R({\cal E}_i)$, the ending node of $e_k$ belongs to $\cal S$, and the length of ${\cal P}_i$ satisfies $k\le r_{\max}-1$. Consider the digraph ${\cal G}_i\doteq ({\cal V},{\cal E}\backslash {\cal E}_i)$. Let $E({\cal P}_i)$ denote the set of all links in ${\cal P}_i$.  If ${\cal E}_i\cap E({\cal P}_i)=\emptyset$, then the path ${\cal P}_i$ still exists in ${\cal G}_i$. In such case, ${\cal E}_i$ and ${\cal E}_j$ are generically distinguishable by Proposition \ref{generic_distinguish}. Otherwise, if ${\cal E}_i\cap E({\cal P}_i)\neq \emptyset$, suppose that ${\cal E}_i\cap E({\cal P}_i)=\{e_{w(1)},...,e_{w(l)}\}$ where $l=|{\cal E}_i\cap E({\cal P}_i)|\le k$, $w(1),...,w(l)\in \{1,...,k\}$, and $w(1)<w(2)<\cdots<w(l)$. Then,  there exists a path $P_s\doteq \{e_{w(l)+1},e_{w(l)+2},...,e_k\}$ in ${\cal G}_i$ with length no more than $k$. Such path starts from $V_R(e_{w(l)})\in V_R({\cal E}_{ij})$ and ends at $\cal S$. Consequently, ${\cal E}_i$ and ${\cal E}_j$ are generically distinguishable by Proposition \ref{generic_distinguish}. Since $i,j$ can be arbitrary, this proves the generic isolability of $\mathbb E$ by Theorem \ref{main_generic_iso}.
\end{proof}}

A commonly discussed failure set is ${\mathbb E}={\cal E}$, i.e, the set of all single-link failures \cite{Rahimian2015DetectionAI}.
The following corollary, immediate from Proposition \ref{equivalent}, points out that sensor placement for generic detectability of every single-link failure is equivalent to that for generic isolability of the set of all single-link failures for the networked system, which is a little surprising.


\begin{corollary} \label{equivalent_singlelink}
In networked system (\ref{sub_dynamic}), if every single-link failure of $\cal G$ is generically detectable for a sensor placement $\cal S$, then the set of all single-link failures (i.e., ${\mathbb E}={\cal E}$) is generically isolable.
\end{corollary}

{
On the other hand, the following corollary validates the intuition that, a failure set not being isolable may arise if it contains two elements, of which one is contained in the other and their difference is not generically detectable in the faultless system.

\begin{corollary}[A class of generically not isolable failure sets] \label{notisolable}
Consider a failure set ${\mathbb E}$ of the networked system (\ref{sub_dynamic}). If ${\mathbb E}$ contains two elements ${\cal E}_i$ and ${\cal E}_j$ such that ${\cal E}_j\subseteq {\cal E}_i$, and ${\cal E}_i\backslash {\cal E}_j$ is generically undetectable, then ${\mathbb E}$ is generically not isolable.
\end{corollary}
\begin{proof} The proof is straightforward from Theorem \ref{main_generic_iso}.
\end{proof}}


\section{Sensor Placement for Generic Detectability and Isolability}\label{section_sensor}
In this section, on the basis of results in Section \ref{fengwei}, we explore the problems of determining the minimum number of sensors to ensure generic detectability and isolability. We will reduce these problems to the hitting set problems and use greedy algorithms to approximate them with guaranteed performances. { It is remarkable that linking the detection of outbreaks (spreading of information) or link failures over networks to the coverage (or connectivity) from a subset of nodes can also be found in the computer community (c.f. \cite{leskovec2007cost-effective,kleinberg2008network}), where the networks are static without nodal dynamics. }
\subsection{Sensor Placement Problems}
We consider two sensor placement problems.
\begin{problem}[sensor placement for detectability of every single-link failure] \label{prob1} For networked system (\ref{sub_dynamic}), determine the minimum number of sensors such that the failure of every single-link of $\cal E$ is generically detectable.
\end{problem}

\begin{problem}[sensor placement for failure isolability]\label{prob2} For networked system (\ref{sub_dynamic}), determine the minimum number of sensors such that a given failure set ${\mathbb E}=\{{\cal E}_1,...,{\cal E}_r\}$ is generically isolable.
\end{problem}

By Corollary \ref{equivalent_singlelink}, if ${\mathbb E}={\cal E}$, Problem \ref{prob2} is exactly equivalent to Problem \ref{prob1}.
\subsection{Hitting Set Problem} Both Problems \ref{prob1} and \ref{prob2} are combinatorial problems.  To further solve them, we introduce the {hitting set} problem.
\begin{definition}[Hitting set problem] Let ${\Sigma}=\{S_1,...,S_q\}$ be a collection of subset of $V$, i.e., $S_i\subseteq V$, $\forall i$.  The hitting set problem is to find the smallest subset $\bar S\subseteq V$ that intersects (hits) every set in $\Sigma$, i.e., $S_i\cap \bar S\ne \emptyset$, $\forall i$.
\end{definition}

Hitting set problem is known to be NP-hard. The greedy algorithm can return a solution with a multiplicative factor $O({\rm ln}\,q)$, more precisely, $1+{\rm ln}\,q$, of the optimal solution, which is the best approximation performance that could be achieved in polynomial time \cite{Submodular}. The greedy algorithm for solving a hitting set problem is given as Algorithm \ref{alg1}, in which the function $f(\bar S)$ is defined as $f(\bar S)=\sum \nolimits_{i=1}^q {\mathbb I}(S_i\cap \bar S)$ for $\bar S\subseteq V$, where function ${\mathbb I}(x)=1$ if $x\ne \emptyset$, otherwise ${\mathbb I}(x)=0$. The basic idea is to find the element from $V\backslash \bar S$ that returns the maximum increase in the number of intersected elements between $\bar S$ and $S_i|_{i=1}^q$ in each iteration.

\begin{algorithm} 
  {{{{
\caption{:Greedy Algorithm for Hitting Set Problem} 
\label{alg1} 
\begin{algorithmic}[1] 
\REQUIRE ($\Sigma$, $V$) 
\STATE Initialize $\bar S=\emptyset$.
\WHILE{$f(\bar S)  < q$}
\STATE $\bar s \leftarrow \arg \mathop {\max }\nolimits_{s \in { V}\setminus \bar S} f(\bar S\cup\{s\})-f(\bar S)$
 \STATE $\bar S \leftarrow {\bar S} \cup \{\bar s\}$
 \ENDWHILE
 \ENSURE $\bar S$
\end{algorithmic}}}
}}
\end{algorithm}

\subsection{Analysis and Algorithms}
An analytical result is first given as follows, which, immediate from Theorem \ref{main_generic_dect}, is the basis of the subsequent derivations.
\begin{proposition} For networked system (\ref{sub_dynamic}), the minimum number of sensors for generic detectability of arbitrary given failure ${\cal E}_f$ is $1$. Moreover, any node in $\bigcup \nolimits_{i\in V_R({\cal E}_f)}\{j\in {\cal V}: {\rm dist}(i,j,{\cal G})\le r_{\max}-1\}$ can be the sensor node.
\end{proposition}


%

Consider Problem \ref{prob1}. Denote the set of nodes which has at least one ingoing link (including self-loop) from other nodes in $\cal G$ by ${\cal V}_s$, i.e., ${\cal V}_s=\{i\in {\cal V}: (j,i)\in {\cal E}, j\in {\cal V}\}$.  For each $i\in {\cal V}_s$, denote the set of nodes whose distance from $i$ is not greater than $r_{\max}-1$ by $S_i$, i.e., $S_i=\{j\in{\cal V}: {\rm dist}(i,j,{\cal G})\le r_{\max}-1\}$. From Theorem \ref{main_generic_dect}, a sensor location ${\cal S}\subseteq {\cal V}$ makes every single-link failure generically detectable, if and only if ${\cal S}$ intersects every $S_i$, i.e., $S_i\cap {\cal S}\ne \emptyset$, $\forall i\in {\cal V}_s$. Let $$\Sigma=\{S_i:i\in {\cal V}_s\}.$$Then, finding the smallest $\cal S$ is equivalent to solving the hitting set problem on $(\Sigma,{\cal V})$. Hence, the greedy algorithm (Algorithm \ref{alg1}) could be adopted to approximate Problem \ref{prob1}.

Consider Problem \ref{prob2}. For each pair $i,j\in\{0,1,...,r\}$ with $i<j$, define a set $S_{ij}\subseteq {\cal V}$ as the set of sensor nodes associated with which $\Phi_i$ and $\Phi_j$ are generically distinguishable. From Corollary \ref{single_failure},
$$S_{ij}=\bigcup\nolimits_{k\in V_R({\cal E}_{ij})}\{l\in {\cal V}: {\rm dist}(k,l,{{\cal G}_i})\le r_{\max}-1\},$$
i.e., $S_{ij}$ is the set of nodes whose distance from one node of $V_R({\cal E}_{ij})$ is no more than $r_{\max}-1$ in ${\cal G}_i$. Afterwards, define a collection $\bar \Sigma$ as
$$\bar \Sigma=\{S_{01},S_{02},\cdots,S_{0r},S_{12},\cdots,S_{1r},\cdots,S_{r-1,r}\}.$$
From Theorem \ref{main_generic_iso}, a sensor location $\bar {\cal S}\subseteq {\cal V}$ makes ${\mathbb E}$ generically isolable, if and only if $\bar {\cal S}$ intersects every set in $\bar \Sigma$. Hence, finding the smallest $\bar {\cal S}$ is equivalent to solving the hitting set problem on $(\bar \Sigma, {\cal V})$, which could also be approximated via the greedy algorithm.

We summarize the above analysis as follows, along with some guaranteed performances of the associated algorithms.

\begin{proposition} \label{prosensor} Problem \ref{prob1} is equivalent to the hitting set problem on $(\Sigma,\cal V)$. The greedy algorithm (Algorithm \ref{alg1}) can return an $O({\rm ln}\,|{\cal V}_s|)$ approximation of the optimal solution.
\end{proposition}

\begin{proposition} \label{prosensor2} Problem \ref{prob2} is equivalent to the hitting set problem on $(\bar \Sigma, \cal V)$. Algorithm \ref{alg1} can return an $O({\rm ln}\, \frac{1}{2}(r+1)r)$ approximation of the optimal solution.
\end{proposition}

 {\begin{remark}
 The approaches in this section provide solutions for the associated sensor placement problems based only on the network topologies with generic weights. In practical scenarios, some real factors may need further considering. One is the sensor resolution or the presence of noise. If the sensor resolution is too low or the signal-to-noise ratio is too small, then the sensors might not be able to distinguish the difference between the noisy nominal output and the faulty one (see \cite{Battistelli2015DetectingTV}). Additionally, if the subsystems have very slow response, or the faulty links have very small weights, then a relatively long observation time may need to distinguish the corresponding outputs.  All these factors will affect the choice of sensor locations in turn (see an example in Section \ref{section_example}). In these scenarios, some quantitive metrics may need to be developed to measure the detection/isolation performances/difficulties (such as the sensitivity of sensors to the effects of faulty links versus noise \cite{Chen1999RobustMF}, the distance between
the faulty trajectory and the nearest nominal one \cite{Baglietto2014DistinguishabilityOD}), just like the controllability metrics in \cite{T2016On}. The corresponding sensor placements  may then be cast as optimization problems to optimize these metrics.\end{remark} 
}
\section{Simulations and Examples}\label{section_example}
We present some simulations and examples to illustrate the main results of this paper.
\subsection{The Five-Node Networked System in Example \ref{example2}}
Consider the five-node networked system in Example \ref{example2}. Let all links shown in Fig. \ref{example_5nodes} have weight $1$. First, in line with Example \ref{example2}, to show the detectability of each single-link failure with sensor node ${\cal S}=\{1\}$, we collect the output responses of the corresponding systems after each single-link failure with a common random initial state $x_0\in {\mathbb R}^{15}$ in Fig. \ref{simulation_exp1}. From this figure, the output response after the failure of link $(1,2),(2,3)$, or $(3,4)$ is the same as that of the original system, whiles the output response after the failure of link $(2,5)$, $(4,5)$, or $(5,1)$ is different from that of the original system, which means each failure of the former three links is undetectable, and the contrary for the latter three links. This is consistent with the claim made in Example \ref{example2} based on Theorem \ref{main_generic_dect}.
Moreover, suppose our goal is to make every single-link failure detectable using as less sensors as possible. According to Proposition \ref{prosensor}, we can construct an equivalent hitting set problem as follows $$\Sigma=\big\{\{1\},\{2\},\{3,5\},\{4\},\{4\},\{5\}\big\}, {\cal V}=\{1,...,5\}.$$
Using Algorithm \ref{alg1} returns ${\cal S}=\{1,2,4,5\}$, which is the optimal solution. 

Now, consider the failure set ${\mathbb E}=\big\{\{(4,5),(3,4)\},\{(4,5)\}\big\}$. The output responses of the resulting systems after each failure with a common random initial state $x_0\in {\mathbb R}^{15}$ are shown in Fig. \ref{simulation_exp2}. From this figure, we know that both $\{(4,5),(3,4)\}$ and $\{(4,5)\}$ are detectable. However, $\mathbb E$ is not isolable because its two elements always generate the same outputs. This is consistent with Corollary \ref{notisolable}.

Finally, suppose our goal is to deployment the smallest sensors so that $\mathbb E$ is isolable. According to Proposition \ref{prosensor2}, this problem is equivalent to the hitting set problem defined as follows $$\bar \Sigma=\{\{1\},\{1,5\},\{4\}\}, {\cal V}=\{1,...,5\}.$$
The greedy algorithm returns $\bar {\cal S}=\{1,4\}$. Through exhaustive search, this solution is optimal. The isolability of $\mathbb E$ is validated by the output responses of the corresponding systems after failures; see Fig. \ref{simulation_exp2_sensor}.

\begin{figure}
  \centering
  \includegraphics[width=2.9in]{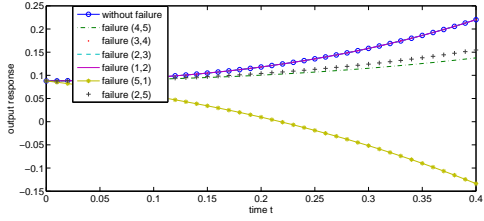}\\
  \caption{Output responses of the networked system in Example \ref{example2} after every single-link failure with ${\cal S}=\{1\}$. }\label{simulation_exp1}
\end{figure}

\begin{figure}
  \centering
  \includegraphics[width=2.9in]{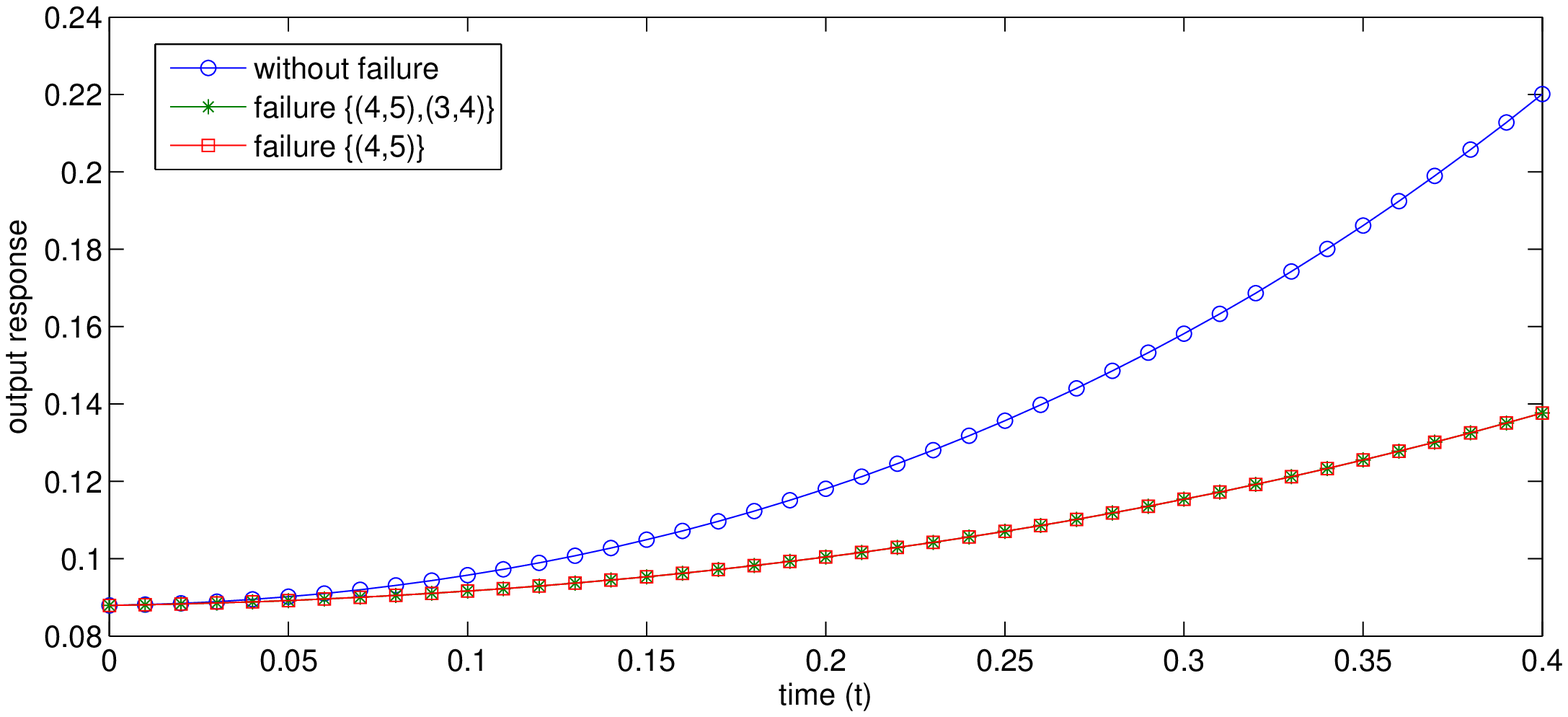}\\
  \caption{Output responses of the networked system in Example \ref{example2} after the failure set ${\mathbb E}=\{\{(4,5),(3,4)\},\{(4,5)\}\}$ with ${\cal S}=\{1\}$. }\label{simulation_exp2}
\end{figure}

\begin{figure}
  \centering
  \includegraphics[width=2.9in]{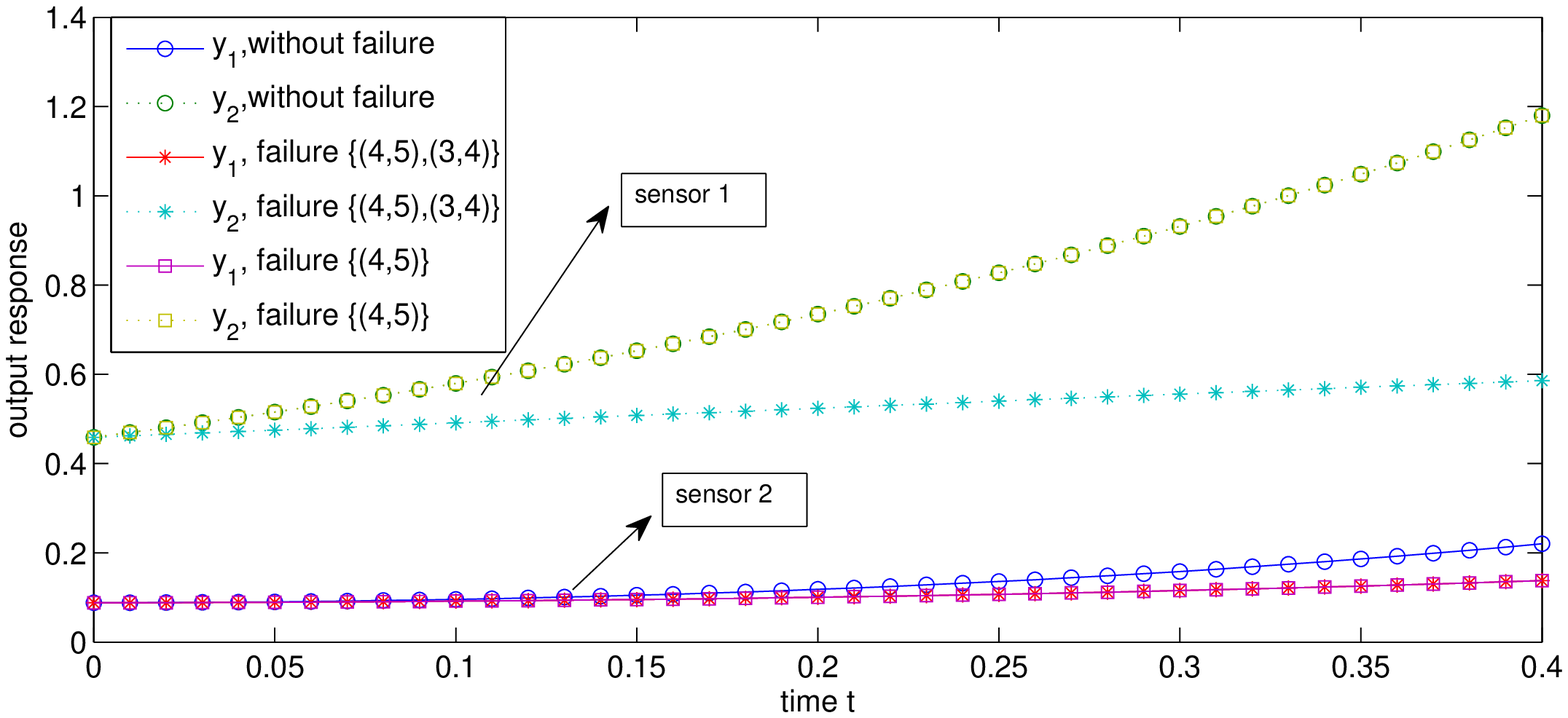}\\
  \caption{Output responses of the networked system in Example \ref{example2} after the failure set ${\mathbb E}=\big\{\{(4,5),(3,4)\},\{(4,5)\}\big\}$ with ${\cal S}=\{1,4\}$. }\label{simulation_exp2_sensor}
\end{figure}


\subsection{Power Network}
Consider a power network consisting of $N$ generators. The dynamics of each generator around its equilibrium state could be described by the following linearized Swing equation \cite{kundur1994power}:
\begin{equation}\label{sub_power} {m_i}{{\ddot \theta }_i} + {d_i}{{\dot \theta }_i} =  - \sum\limits_{j=1}^N {{k_{ij}}{\rm{(}}{\theta _i}{\rm{ - }}{\theta _j}{\rm{)}}}, \ y_i=\theta_i,\end{equation}
$i\in\{1,...,N\}$, where $\theta_i$ is the phrase angle, $m_i$ and $d_i$ are respectively the inertia
and damping coefficients, and $k_{ij}$ is the susceptance of the power line from the $j$th generator to the $i$th one. Rewrite (\ref{sub_power}) as
\begin{equation}\label{sub_power2} \begin{array}{l}
\left[ {\begin{array}{*{20}{c}}
{{{\dot \theta }_i}}\\
{{{\ddot \theta }_i}}
\end{array}} \right] = {\left[ {\begin{array}{*{20}{c}}
0&1\\
0&{{{ - {d_i}}}/{{{m_i}}}}
\end{array}} \right]}\left[ {\begin{array}{*{20}{c}}
{{\theta _i}}\\
{{{\dot \theta }_i}}
\end{array}} \right] + {\left[ {\begin{array}{*{20}{c}}
0\\
1
\end{array}} \right]}\sum\limits_{j = 1}^N w_{ij} [1,0]\left[\!\! {\begin{array}{*{20}{c}}
{{\theta _j}}\\
{{{\dot \theta }_j}}
\end{array}}\!\! \right],\\ y_i=[1,0]\left[\!\! {\begin{array}{*{20}{c}}
{{\theta _j}}\\
{{{\dot \theta }_j}}
\end{array}}\!\! \right],
\end{array}\end{equation}where $w_{ij}=k_{ij}/{m_i}$ if $j\ne i$, and $w_{ii}=-\sum\nolimits_{j=1,j\ne i}^Nk_{ij}/{m_i}$, which can be seen as weight of the self-loop $(i,i)$. A typical power network topology is the IEEE-9 bus system shown in Fig. \ref{ieee9bus}, which consists of $9$ buses and whose link set is denoted by $\cal E$. In our analysis, each bus is simplified as a generator \cite{Kumar2012TransientSA}.

Consider the failure of one bus from the IEEE-9 bus system. For example, suppose that bus $1$ is removed from this power network, i.e., ${\cal E}_f=\{(1,4),(4,1),(1,1)\}$ (it should be noted that, the influence on the self-loops of other nodes from the removal of bus $1$ is neglected). It can be seen that, $r_{\max}=\infty$ for the dynamics (\ref{sub_power2}) whatever value $\frac{-d_i}{m_i}$ takes.  According to Theorem \ref{main_generic_dect}, deploying one sensor on an arbitrary bus can detect this failure.

Furthermore, suppose we have the prior knowledge that at most one bus is removed from the power network. Then, in this situation the failure set can be formulated as ${\mathbb E}=\{{\cal E}_{fi}|_{i=1}^9\}$, where ${\cal E}_{fi}\doteq \{(i,j):(i,j)\in {\cal E} \}\bigcup \{(j,i):(j,i)
\in {\cal E}\}$, i.e., ${\cal E}_{fi}$ collects all ingoing and outgoing links of node $i$. By the greedy algorithm described in Algorithm \ref{alg1}, a sensor placement solution is obtained as $\bar {\cal S}=\{4\}$ (in fact, deploying one sensor at an arbitrary bus is feasible for failure isolability). Letting ${-d_i}/{m_i}=-1$, $\forall i$, and $w_{ij}=1$ for any links except the self-loops, we collect in Fig. \ref{ieee9busResponse} the output responses of the corresponding systems after every single-node failure with a common random initial state $x_0$.  It validates that, indeed,  the set of every single-node failure is isolable by the proposed sensor deployment. 

{Finally, consider the scenario where each sensor is affected by a scalar white noise with zero mean and a standard deviation of $0.05$. Suppose that for two vectors $y(t)$ and $\bar y(t)$, if $||y(t)-\bar y(t)||_2\le E$ then $y(t)$ and $\bar y(t)$ cannot be distinguished by the sensors (and otherwise can), where $||\cdot ||_2$ takes the  $2$-norm, and $E$ is a prescribed threshold (for simplifying descriptions, $E$ does not vary with the number of sensors). Consider the removal of bus $1$, with three different sensor locations the first being ${\cal S}_1=\{4\}$, the second ${\cal S}_2=\{4,3\}$, and the third ${\cal S}_3=\{4,5,3\}$. The remaining system parameters are the same as those mentioned above. Fig. \ref{ieee9busResponseNoise} records the raw output deviations from the nominal measured one (without filtering) as well as decisions made by the sensors over the time axis. It can be seen that, though all sensor solutions can detect the failure, the solution with more sensors achieves a shorter observation time. Besides, the sensor resolution (the threshold $E$) also affects the timeliness of failure detection. These observations indicate that some further real factors may need taking into account in practical implementations apart from the generic detectability.}

\begin{figure}
  \centering
  \includegraphics[width=1.6in]{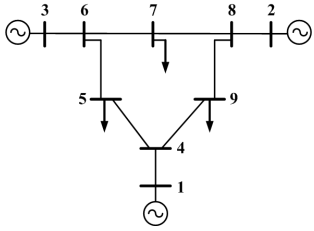}\\
  \caption{Sketch of the IEEE-9 bus system \cite{Kumar2012TransientSA}. Every link is bidirectional. }\label{ieee9bus}
\end{figure}

\begin{figure}
  \centering
  \includegraphics[width=3.0in]{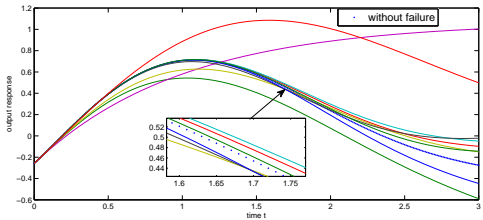}\\
  \caption{Output responses of the IEEE-9 bus power system after every
single-node failure with ${\cal S}=\{4\}$.}\label{ieee9busResponse}
\end{figure}

\begin{figure}
  \centering
  \includegraphics[width=3.3in]{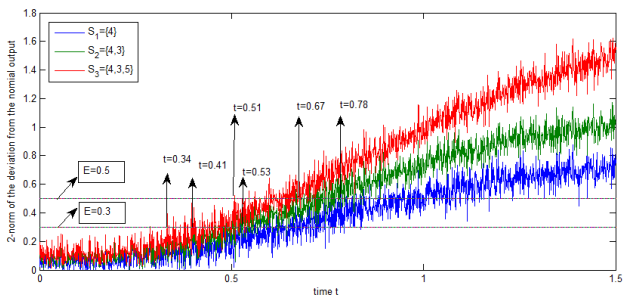}\\
  \caption{Deviations of outputs from the nominal one of the IEEE-9 bus power system after the
  removal of bus $1$ in the presence of measurement noise. The numbers around the arrows are the time that sensors first distinguish the faulty outputs from the nominal one with the threshold $E$.}\label{ieee9busResponseNoise}
\end{figure}

\section{Conclusions}
In this paper, we study generic detectability and isolability of topology failures for a networked linear system, where subsystem dynamics are given and identical, but the weights of interaction links among them are unknown. We give necessary and sufficient graph-theoretical conditions for generic detectability and isolability, as well as some characterizations of generically (not) isolable failure sets. These conditions reveal fundamental structural/topological limitations for the networked systems to support detectability and isolability of a given topology failure (set), which are irrespective of the exact detection and isolation algorithms adopted. These results are further used to deploy the smallest set of sensors to achieve generic detectability and isolability of a given failure (set).

{We summarize some future research directions here concerning on the practical limitations of our results. One is a more reasonable (possibly dynamical \cite{Buldyrev2009Catastrophic}) model of a topology failure which should distinguish its affections from those of (measurement or process) noise or small parameter perturbations. The second is a quantitive metric which can measure the associated detection/isolation performances/difficulties, and may take some real factors such as the measurement noise, the sensor resolutions, the initial states or the observation time into account. The final one is developing the exact detection and isolation algorithms. Since applying some existing observer-based approaches for the lumped systems requires accurate system parameters \cite{Chen1999RobustMF},  it is of great value to explore some data-driven approaches without system identification, which will be our future work.}

{\bibliographystyle{elsarticle-num}
{\footnotesize
\bibliography{yuanz3}
}}

\end{document}